\documentclass[11pt]{amsart}
\usepackage{latexsym, amssymb, amsmath, color, enumerate, graphicx, slashbox}
\usepackage{graphicx}
\usepackage{cite}
\usepackage{verbatim}

\makeatletter

\newcommand{\Rmnum}[1]{\expandafter\@slowromancap\romannumeral #1@}
\makeatother

\newcommand{\RR}{{\mathbb R}}
\newcommand{\NN}{{\mathbb N}}

\newcommand{\cE}{{\mathcal E}}

\newcommand{\cG}{{\mathcal G}}

\newcommand{\cI}{{\mathcal I}}

\newcommand{\cN}{{\mathcal N}}

\newcommand{\cR}{{\mathcal R}}

\newcommand{\cU}{{\mathcal U}}

\newcommand{\D}{{\partial }}

\newcommand\aonesq{a'{}^2}
\newcommand\atwo{a'{}'}

\newcommand\ie{{\em i.e., }}

\newtheorem{theorem}{Theorem}[section]

\newtheorem{lemma}[theorem]{Lemma}
\newtheorem{definition}[theorem]{Definition}

\newtheorem{example}[theorem]{Example}
\theoremstyle{remark}
\newtheorem{remark}[theorem]{Remark}

\numberwithin{equation}{section}

\newcommand{\beq}{\begin{equation}}
\newcommand{\eeq}{\end{equation}}

\renewcommand{\AA}{\ensuremath{\mathbb{A}}} 

\renewcommand{\a}{\alpha}

\newcommand{\s}{\sigma}

\renewcommand{\t}{\tau}

\newcommand{\NAME}{Dyson-Taylor commutator method}

\definecolor{DarkRed}{rgb}{0.8,0.3,0.6}

\definecolor{DarkerBlue}{rgb}{0.1,0.3,1}

\definecolor{Green}{rgb}{0.2,0.7,0.3}

\author{Wen Cheng}
\email{cheng@math.psu.edu}
\address{Department of Mathematics, Pennsylvania State University, University Park, PA 16802}

\author{Nick Costanzino}
\email{costanzi@math.psu.edu}
\address{Department of Mathematics, Pennsylvania State University, University Park, PA 16802}

\author{John Liechty}
\email{jcl12@psu.edu}
\address{Department of
Marketing, Smeal College of Business, and Department of Statistics,
Pennsylvania State University, University Park, PA 16802}

\author{Anna Mazzucato}
\email{mazzucat@math.psu.edu}
\address{Department of Mathematics, Pennsylvania State University, University Park, PA 16802}

\author{Victor Nistor}
\email{nistor@math.psu.edu}
\address{Department of Mathematics, Pennsylvania State University, University Park, PA 16802}

\thanks{A.M. was partially supported by NSF Grant DMS 0708902. V.N. was partially
supported by NSF grant DMS-0555831, DMS-0713743, and OCI 0749202.}
\date{\today}

\begin{document}

\title[Closed-form asymptotics for local volatility
  models]{Closed-form asymptotics for local volatility models}

\begin{abstract}
We obtain new closed-form pricing formulas for contingent claims when
the asset follows a Dupire-type local volatility model. To obtain the
formulas we use the \NAME\ that we have recently developed in
\cite{CCMN,CCCMN,CMN} for short-time asymptotic expansions of heat
kernels, and obtain a family of general closed-form approximate
solutions for both the pricing kernel and derivative price. A
bootstrap scheme allows us to extend our method to large time. We also
perform analytic as well as a numerical error analysis, and compare
our results to other known methods.
\end{abstract}

\maketitle


\tableofcontents

\section{Introduction}
Financial derivatives (also known as contingent claims) are now a
ubiquitous tool in risk management with approximately 600 trillion
dollars worth of such contracts currently in the market. The pricing
of such derivatives is therefore an active area of research in both
Mathematics and Finance (see for example
\cite{Duf,FPS,Gath,Hull,Shreve} and the references therein).  In this
paper, we will apply the perturbative (asymptotic) method introduced
in \cite{CCMN} for numerically solving parabolic equations and
  then use this method to price European options.

One of the earliest models used in pricing derivatives is the
Black-Scholes-Merton model \cite{BS,Mer}, for which the movement in
the price $X_t$ of the underlying asset on which the claim is based is
modeled by geometric Brownian motion.  For the Black-Scholes-Merton as
well as for other models given by stochastic differential equations,
the pricing of European options can be reduced to the calculation of
certain solutions of parabolic equations, obtained through Ito's Lemma
(and the change of variables $t \leftarrow T-t$) in the backward
Kolmogorov equation. The resulting equation is a Fokker-Planck
equation, which is an equation of parabolic type. Fokker-Planck
equations more generally have important applications in statistical
mechanics and in probability (see for example the monographs
\cite{Risken,Carmichael,Gardiner}). Given that the asset price is
always assumed positive, the Fokker-Planck equation is solved on the
positive half-line. One difficulty in treating this type of the
equation is that the coefficients of the Fokker--Planck operator
typically vanish at the boundary, making the equation degenerate.

For example, for the Black-Scholes-Merton model, the
resulting Fokker-Planck equation is given by
\begin{equation}\label{eq.U}
\begin{cases}
  \D_t U(t,x) - L U(t,x) = 0, & 0<t<T, \, x>0 \\
  U(0,x) =  h(x), & x>0,
\end{cases}
\end{equation}
where
\begin{align}\label{eq.LBSM}
  L := \frac{1}{2} \sigma^2 x^2 \D_x^2 + r x\D_x - r,
\end{align}
is the {\em Black-Scholes operator}, a degenerate elliptic operator,
$t$ is the time to expiry, and $h$ is the so-called {\em pay-off}
function. For a {\em European Call option with strike $K$} and expiry
(or exercice) date $T$, the pay-off function $h$ is given by the
formula $h(x_T) = |x_T-K|_{+} := \max\{x_T-K, 0\}$, where $x_T$ is the
price of the underlying asset at time $T$. Above, $\sigma$ and $r$ are
constant parameters, representing respectively the {\em volatility} of
the underlying asset, and the current {\em interest rate}. Since the
operator is degenerate at the boundary $x=0$, it can be shown that the
solution automatically vanishes there and no explicit boundary
condition need to be imposed.

A popular model related to the Black-Scholes-Merton model is the CEV
model \cite{CR}. In the CEV model, the operator $L$ is the form
\begin{equation}\label{eq.LCEV}
  L(t,x)=L(x)=\frac{1}{2}\sigma^2 x^{2\alpha} \D_x^2+rx\D_x-r,
\end{equation}
where $\sigma$, $\alpha$, $r$ are constant.  Yet another popular model
is Dupire's local volatility model, for which we allow the volatility
to change with time:
\begin{equation*}
  L(t,x)=L(x)=\frac{1}{2}\sigma^2(x,t) x^{2} \D_x^2+rx\D_x-r.
\end{equation*}

Except in special cases, such as the Black-Scholes-Merton equation above and
when $L$ has constant coefficients, very few exact solution formulas
to the problem \eqref{eq.U} are available.  It is therefore important
to devise fast, accurate approximate solution methods.  The focus of
this paper is on obtaining approximate solution methods that are fast
and accurate by combining standard numerical methods with the
asymptotic techniques developped in \cite{CCMN}.  Fast solution
methods are crucial when calibrating unknown parameters, especially in
the Baeysian inference framework. We hope to address this
  question in a forthcoming paper.

In view of the above discussion, it is justified to study the forward
initial-value problem \eqref{eq.U} for the general case when $L$ is an
operator of the form:
\begin{align}\label{eq.Lgen}
  L(t) := \frac{1}{2} a(t,x)^2 \D_x^2 + b(t,x) \D_x + c(t,x).
\end{align}
We therefore allow for {\em variable} coefficients {\em in both space
  and time}. We assume throughout that $a(x)>0$, for $x>0$ and that
the coefficients $a$, $b$, $c$ are smooth functions.  The perturbative
method introduced in \cite{CCMN} for the study of parabolic equations
in arbitrary dimensions was fully justified in the case when $a$, $b$,
and $c$ and all their derivatives are bounded, and are bounded away
from zero:\ $a(x) \geq \gamma >0$. In this paper we complete the
results of \cite{CCMN} with explicit formulas for the 1D case. Then we
numerically test our formulas for the Black-Scholes-Merton and CEV
models, obtaining an excellent agreement between our theoretical
results and the numerical tests. Both the Black-Scholes-Merton model
\eqref{eq.LBSM} and and the CEV model \eqref{eq.LCEV} are more general
than the models considered in \cite{CCMN} in that their coefficients
do not satisfy the assumptions of the paper, yet the numerical tests
indicates that the results of that paper are still valid for the more
general models considered here.  This observation suggests that the
theoretical framework of \cite{CCMN} is applicable in greater
generality. We plan to study this point in a forthcoming paper.

To explain our method, let us recall that, under certain conditions on
the operator $L$ and initial value $h$, described in details in the
next section, there exists a {\em smooth function} $\cG_t(x,y)$ such
that the solution to \eqref{eq.U} has the representation
\begin{align}\label{eq.convolution.U}
  U(t,x) = \int_0^\infty \cG_t(x,y) h(y) dy.
\end{align}
The kernel function $\cG_t(x,y)$ in \eqref{eq.convolution.U} is the
fundamental solution or the so-called {\em Green function} for the
problem \eqref{eq.U}.

\begin{remark} \label{rem.Kernel}
Given that $\cG_t(x,y)$ arises in several different contexts, we will
call the function $\cG_t(x,y)$ the {\em transition density kernel,
pricing kernel, heat kernel}, or {\em Green function}
interchangeably, depending on the context in which the object arises.
\end{remark}

As mentioned above, except for some very special cases no explicit
formulas for $\cG_t(x,y)$ or $U(t,x)$ are available.  For the
Black-Scholes-Merton model, a change of variables reduces the PDE to a
heat equation that can then be solved explicitly.  Therefore, exact
formulas for the kernel $\cG^{BSM}$ and the solution $U^{BSM}$ exist,
which we recall now for further reference:
\begin{align}\label{BSM}
\begin{split}
& \cG^{BSM}_t(x,y) = \frac{\exp(-r\, t)}{y \sqrt{2 \pi \s^2 t}}
  \exp \left(- \frac{ | \ln(x/y) + (r - \sigma^2/2) t |^2}
       {2\sigma^2 t} \right) \\ & U^{BSM}(t,x) = \int_0^\infty
       \cG^{BSM}_t(x,y) dy = x \cN(d_-) + K e^{-r\, t}\cN(d_+), \\
\end{split}
\end{align}
where $\cN(x) = \int_{-\infty}^x \frac{1}{\sqrt{2\pi}} e^{-z^2/2}$ is
the cumulative normal distribution function (cumulative Gaussian
distribution function) and
\begin{align}
  d_\pm = \frac{\ln(x/K) + (r \pm \sigma^2/2) t}{\sigma\sqrt{t}}.
\end{align}
However, for the time-dependent Black-Scholes-Merton model, where
$\sigma$ and $r$ are time-dependent, or local volatility models in
general, closed form solutions are generally given by series
expansions and difficult to use in practice or are not known (see, for
instance, \cite{DSHB, LY}).

The method that we use in this paper is to give an approximate
closed-form solution for the equation \eqref{eq.U} by giving an
approximate closed-form expansion for the Green's function
$\cG_t(x,y)$. Since our approximation of the Green's function is in
terms of Gaussian-type integrals, it gives a closed-form for the
approximate price of a European call option for {\em any}
one-dimensional model where the operator $L$ is given by
\eqref{eq.Lgen}. In fact, as an application, we give the prices and
Greeks (that is, suitable derivatives) of a European call option and
perform an error analysis in Section \ref{sec.error}.

There exists a vast literature on obtaining asymptotic expansions of
the Green's function $\cG_t(x,y)$ when $t$ small and $x$ is close to
$y$, especially in the case that $L$ is independent of time \cite{Az,
  Hsu, Kampen, McKeanSinger,Pleijel, V1, V1, Vas}. (See also
\cite{Ait,Farkas, Greiner, Melrose2, TayPDEII}). Many of these methods
are based on a geometric interpretation of the operator $L$ (or at
least its principal part) as a Laplace operator on curved space, and
require computing the geodesics in this space, which very often must
be done numerically. Other approaches are based on pseudo-differential
calculus.  In particular, Corielli, Foschi, and Pascucci \cite{CFP}
use a {\em parametrix} construction for the problem \eqref{eq.U} to
obtain a closed-form approximate solution.  We recently developed in
\cite{CCMN,CCCMN,CMN} a complementary approach to computing short-time
asymptotics for $\cG_t$, based on parabolic rescaling, Taylor's
expansions of the coefficients, Duhamel's and Dyson's formulas,
and exact commutator expansion. We called this method the {\em \NAME}.
Our method is more elementary and appears very stable in practical
implementations.

Let us fix a function $z = z(x,y)$ with the properties that $z(x,x)=x$
and all its derivatives are bounded. The function $z$ will represent
the {\em basepoint} for a parabolic rescaling of the Green's function.
Then our short-time asymptotics give an expansion for the kernel in
the form:
\begin{align}\label{eq.G.expansion}
\begin{split}
  \cG_t(x,y) &= \cG_t^{[n]}(x,y;z) +t^{(n+1)/2}
  \cE_t^{[n]}(x,y; z) \\
  \cG_t^{[n]}(x,y;z) &= G_t^{[0]}(x,y ; z ) + t^{1/2}
  G_t^{[1]}(x,y ; z ) + \cdots + t^{n/2} G^{[n]}_t(x,y ; z ), \\
\end{split}
\end{align}
where $\cG_t^{[n]}$ is the sum on the first $n$ terms of the expansion
and represents the $n$-th order approximate kernel, while
$t^{(n+1)/2}\,\cE_t^{[n]}$ is the remainder.  The first term,
$G_{t}^{[0]}$ is given by a dilated Gaussian function
\begin{equation}\label{eq.G}
  G_{t}^{[0]}(x, y) =   G_{t}^{[0]}(x-y) =
  \frac{1}{\sqrt{2 t \pi a(0,z)^2 }} \exp( - \frac{|x-y|^2}{2 t a(0,z)^2}).
\end{equation}
The \NAME\ yelds an explicit algorithm to compute the terms $G_t^{[n]}$
for {\em any} \ $n$, if $L$ is an operator of the form \eqref{eq.Lgen}
and corresponding analogs in higher dimension.

More precisely, our main result in \cite{CCMN} is that for the local
volatility operator \eqref{eq.Lgen}, the $n$-th order approximate
kernel has the form

\begin{equation} \label{eq.expansion}
   G_t^{[k]}(x,y) := t^{-1/2}
   \mathfrak{P}^k(\,z, \,z+\, \frac{x-z}{t^{1/2}}\,,
   \,\frac{x-y}{t^{1/2}}\,) G^{[0]}_t(   \,\frac{x-y}{t^{1/2}}\,),
\end{equation}

where the functions $\mathfrak{P}^\ell(z,x,y)$ are algorithmically
computable (recall that $z=z(x,y)$). In this paper we shall compute
the functions $\mathfrak{P}^k$, for $k=0,1,2$ at an {\em arbitrary}
basepoint $z$. The details, based on the \NAME\ method, can be found
in Section \ref{sec.framework} and \ref{sec.basepoints}. We therefore
obtain {\em new} closed form asymptotic expansions of the Green
function for local volatility models. In particular, the first order
asymptotic expansion at arbitrary $z=z(x,y)$ is given by
\begin{multline} \label{eq.main.example}
  \cG_t^{[1]}(x,y;z) = \frac{1}{\sqrt{2 \pi\, t\, a(0,z)^2 }} \left[ 1
    + \frac{3 a(0,z) a '(0,z) - 2 b(0,z) }{2 a(0,z)^2} (x-y)
    \right. \\
    \left.  - \frac{a'(0,z)}{2 t\, a(0,z)^3}(x-y)^3 + (x-z)
    \frac{(x-y)^2 - t\,a(0,z)^2}{t\, a(0,z)^3} \right]
  e^{ - \frac{|x-y|^2}{2 t\, a(0,z)^2}. }
\end{multline}
We provide an explicit formula for the second order expansion of the
Green function at the end of Section \ref{sec.framework}.  This
algorithm can be implemented very efficiently at least in dimension 1
and for $n$ small, $n=1$, $n=2$. The numerical tests in Section
\ref{sec.error} show that already the second-order approximation is
adequate for the Black-Scholes and CEV models.

For each term $G_t^{[k]}$ in the expansion of the Green function, let
$\cU^{[k]}$ denote the corresponding term in the expansion of the
solution,
\begin{align}\label{eq.U.formula}
  \cU^{[k]}(t, x) = \int_0^\infty G_t^{[k]}(x,y) h(y;K)dy.
\end{align}
Then using \eqref{eq.convolution.U} and \eqref{eq.G.expansion}, we
arrive at the expansion of the value of the contingent claim,

\begin{align}\label{eq.U.expansion}
\begin{split}
  & U(t,x) = U^{[n]}(t,x) +t^{(n+1)/2}{\mathfrak E}^{[n]}(t,x) \\ &
  U^{[n]}(t,x) = \cU^{[0]}(t,x) + t^{1/2} \cU^{[1]}(t,x) + t
  \cU^{[2]}(t,x) + \cdots t^{n/2} \cU^{[n]}(t,x)
  \end{split}
\end{align}
where
\begin{align}
  t^{(n+1)/2} {\mathfrak E}^{[n]}(t,x) := t^{(n+1)/2}\int_0^\infty
  \cE_t ^{[n]}h(y;K) = U(t,x) - U^{[n]}(t,x)
\end{align}

is the remainder term (or error) in the expansion of the solution.  In
\cite{CCMN} we have shown that the remainder can be controlled in
exponentially weighted Sobolev norms, when the operator $L$ is
uniformly strongly elliptic.  These bounds on the remainder imply
that, in this case, the error made by replacing $\cG_t$ with
$\cG^{[n]}_t$ in \eqref{eq.convolution.U} is of order $t^{n/2}$ {\em
  globally\/} in space, the expected optimal rate.  In \cite{CCCMN},
we consider degenerate operators, the symbol of which is strongly
elliptic with respect to some complete metric of bounded geometry.
For example, the Black-Scholes and the SABR models fit into this
framework. By contrast, the CEV model with $0<\beta<1$ does not fit
into this framework.  Our numerical tests indicate nevertheless that
the error term has the same order in $t$ even for the CEV model with
$\beta<1$.  For pedagogical purposes and error analysis we will list
all the details for the time-dependent Black-Scholes and CEV models,
although our results are more general.

In Section \ref{sec.error} we perform a numerical error analysis by
computing both the numerical solution $U$ and expansion $U^{[n]}$ and
estimating the error
\begin{align}
    | U(t,x) - U^{[n]}(t,x) |
\end{align}
pointwise for the basepoint
$z(x,y)=x$, when $n=1,2$. The error analysis is in good agreement with
the theoretical results, even though the local volatility operators
considered in this paper do not necessarily satisfy the assumptions on
the coefficients of $L$ needed to establish the analytic error
estimates performed in \cite{CCMN,CCCMN,CMN}.

In Section \ref{sec.error} we then perform an error analysis. For the
Black-Scholes-Merton model, for which an exact solution formula is
readily available, we compare the expansions at the basepoint
$z(x,y)=x$ with the exact solution. (Note however, that numerical
errors arise also in the calculation of exact solutions, due to
round-off errors and other approximations.) For the CEV model, we
compare the expansions with benchmark formulas in the literature, in
particular the Hagan-Woodward implied volatility approximation
\cite{HW}.

Given that the kernel approximation is asymptotic in time, it
guarantees good error control {\em a piori} only for sufficiently
small $t$.  In Section \ref{bootstrap}, we shall introduce a
bootstrap scheme to extend our method to arbitrary large time. This
strategy is based on the evolutionary property of the solution
operator to \eqref{eq.U}. By doing so, we show that the error is
remarkedly reduced. As an application in portfolio management, we also
compute the Greeks (or hedging parameters) of a European call option
and compare our approximations with the true Black-Scholes Greeks in
Section \ref{greeks} and Section \ref{bootstrap}. These
applications again underline the accuracy of our methods.

\subsection*{Acknowledgements}
The authors would like to thank Marco Avellaneda for valuable
suggestions and comments on the manuscript, and Jim Gatheral for
useful discussions. Victor Nistor also gladly acknowledges support
from the Max Planck Institute for Mathematics, where part of this work
has been performed.

\section{Theoretical Framework} \label{sec.framework}
We begin by recalling the \NAME, which we introduced in
\cite{CCMN,CMN}, to obtain small-time asymptotic expansions for the
solution of the initial-value problem:
\begin{align}\label{eq.U.tau}
\begin{split}
& \D_t U(t,x) - L(t) U(t,x) = 0 \\
& U(0,x) = h(x). \\
\end{split}
\end{align}
Throughout the paper, the operator $L$ will be given by
\eqref{eq.Lgen}, and we will omit the explicit dependence of $L$ and
of its coefficients on $x$. In addition, we tacitly assume that all
the coefficients of $L$ are regular enough to carry our the
manipulations described next. For a rigorous justification in the case
$L$ is not degenerate, we refer to \cite{CCMN,CMN}.

If there is a unique solution to the initial-value problem
\eqref{eq.U.tau}, then the linear operator that maps the initial data
$h$ to the solution $U$ is well defined. We refer to such operator as
the {\em solution operator}.  For constant-coefficient second-order
operators, $L_0$, the solution operator forms a {\em semigroup},
denoted by $e^{t \, L_0}$, $t>0$; that is, the solution operator has
the following properties:
\begin{enumerate} \renewcommand{\labelenumi}{(\textit{\roman{enumi}})}
 \item $e^{t\, L_0}\vert_0 =I$.
 \item $e^{t_1\, L_0} \, e^{ t_2\, L_0} = e^{(t_1 + t_2)\, L_0}$, $t_1,
 t_2>0$. \label{i.semigroup2}
\end{enumerate}
The same conclusion hold for variable-coefficient, but time-independet
operators $L$, under some conditions, for instance if $L$ is strongly
elliptic \cite{Pazy} (that is, $a(x)\geq \gamma >0$ for all $x$).
When $L$ is a time-dependent operator, $L=L(t)$, the solution operator
is no more a semigroup, but under some additional mild conditions,
forms an {\em evolution system} $\mathcal{S}(t_1,t_2)$
\cite{Lunardi,CMN}.  For an evolution system, property
(\ref{i.semigroup2}) is replaced by \ $\mathcal{S}(t_1,t_2)
\mathcal{S}(t_2,t_3) = \mathcal{S}(t_1,t_3)$, if $0\leq t_3 \leq t_2
\leq t_1$.  Following the notation set forth in the Introduction, we
denote the kernel or Green's function of the solution operator to the
problem \eqref{eq.U.tau} by $\cG^L_t$.

Our method relies heavily on the study of distribution kernels of the
evolution operators defined by our Fokker-Planck operator, so a brief
discussion of distribution kernels and of our conventions is in order.

\begin{remark} \label{rem.genkerneldef}
Given a linear operator $T$ mapping smooth functions with compact
support into distributions, there exist a {\em distribution kernel}
\ $k_T$ such that
\begin{equation}  \label{eq.genkerneldef}
     T\, u (x) = \int k_T(x,y)\, u(y) \, dy.
\end{equation}
The integral above is interpreted as the pairing between test
functions and distributions.  In this paper, we will be interested in
the integral representation \eqref{eq.genkerneldef} in the case that
$T$ is a {\em smoothing} operator, that is, an operator that maps
compactly supported distributions into smooth functions. Then, the
kernel $k_T$ is a {\em smooth function}, and the notation $k_T(x,y)$
is justified pointwise. (For a more detailed dicussion, see for
example \cite{TayPsdo}.)  In this case, we will write $T(x,y)$ to
denote the kernel $k_T(x,y)$, and in general, {\em we shall identify
  an operator with its distribution kernel}. Let $f$ be a smooth
function, then we denote the operators of multiplication by $f$ also
with $f$.  Additonally, we notice that there is no confusion when
writing $fT$ or $Tf$ since the distribution kernels of these operators
are given by $fT(x,y) = f(x)T(x,y)$ or $Tf(x,y) = T(x,y) f(y)$.
Similarly, there is no confusion when writing $\partial_x T(x,y)$,
since the distribution kernel of $\partial_x T$ is $(\partial_x k_T)
(x,y) = \partial_x (k_T(x,y))$. However $k_{T \partial_x} (x,y) =
-\partial_y k_T(x,y)$.
\end{remark}

We now introduce parabolic rescaling, which is a basic tool used in
this paper. Let $z$ be a fixed, but arbitrary point in $\RR$ and $s>0$
a parameter. Given a function $f(t,x)$ we denote by
\begin{equation}\label{eq.rescaled}
  f^{s,z}(t,x) := f(s^2t, z + s(x - z)),
\end{equation}
the {\em parabolic rescaling by $s$ of the function $f$ about
  $(0,z)$}. Thus $h^{s,z}(x) := h(z+s(x-z))$ for a function that does
not depend on $t$.  We will refer to $z$ as the {\em basepoint} for
the rescaling.  Similarly, we define a rescaled operator $L^{s,z}$ by
\begin{align}\label{eq.L^s}
  L^{s,z}(t,x) := \frac{1}{2} a^{s,z}(t,x)^2 \D_x^2 + s
  b^{s, z}(t,x) \D_x + s^2 c^{s, z}(t,x).
\end{align}
If $U$ solves the initial-value problem \eqref{eq.U.tau}, then
$U^{s,z}$ solves the rescaled problem
\begin{align}\label{eq.U.tau.s}
\begin{split}
  & \partial_t U^{s,z}(t,x) - L^{s,z}(t,x) U^{s,z}(t,x) = 0 \\
  & U^{s,z}(0,x) = h^{s,z}(x)
\end{split}
\end{align}
Consequently, the Green functions of the operator $\partial_t - L$ and
of the rescaled operator $\partial_t- L^{s,z}$ are related by
\begin{equation} \label{eq.correspondence}
\begin{aligned}
\cG_t^{L}(x,y) &= s^{-1} \cG^{L^{s,z}}_{s^{-2}t}(z + s^{-1}(x - z), z +
s^{-1}(y-z))\\ &=t^{-\frac{1}{2}}\cG_1^{L^{\sqrt{t},z}}(z +
t^{-\frac{1}{2}}(x - z), z + t^{-\frac{1}{2}}(y-z)), \text{ if } s =
t^{\frac{1}{2}}.
\end{aligned}
\end{equation}

We now proceed to compute the Green's function $G_t^{L^{s,z}}$ of the
rescaled problem \eqref{eq.U.tau.s} when $t =1$. In order to do so,
we shall consider the Taylor expansion in $s$ at $s=0$ of the rescaled
operator $L^{s,z}$, given in equation \eqref{eq.L^s}, up to order $n$.
By ``Taylor expansion'' we mean that we Taylor expand the coefficients
of $L^{s,z}$ and group all terms of the same order in
$s$. The operator $L^{s,z}$ can then be written as follows
\begin{equation}\label{eq.Taylor}
    L^{s,z} = \sum_{k=0}^{n} s^k L_k^{z} +
      s^{n+1} L_{n+1}^{s,z}(t, x),
\end{equation}
where $L_{n+1}^{s,z}(t, x)$ contains all the remainder terms from the Taylor expansion
of the coefficients.

In this paper, we concentrate on calculating explicitly the
second-order approximation of the Green function of $L$.  Hence, {\em
  we fix $n=2$ from now on}.  For notational convenience, we denote
$g'(t,x) = \frac{\partial}{\partial x} g(t,x)$ and $\dot{g}(t,x) =
\frac{\partial}{\partial t} g(t,x)$. Then the second-order Taylor
expansion in $s$ of $f^{s,z}$ at $s=0$ is given by
\begin{multline}\label{eq.Taylor.rescaled}
  f^{s,z}(t, x) = f(0, z) + s(x-z)f'(0,z)\\
  + s^2 t \dot{f}(0, z)
  + s^2 (x-z)^2 f''(0, z)/2 + s^3 r(s, t, x, z),
\end{multline}
with $s^3 r(s, t, x, z)$ the remainder.
Below $a = a(0,z)$ and all the other functions are to be evaluated
at $(0,z)$, unless stated otherwise.
We then readily have the
second order Taylor expansion of $L^{s,z}$ in $s$ at $s=0$:
\begin{equation} \label{eq.Taylor.L}
  L_0^{z} := \frac{1}{2} a^2 \partial_x^2, \quad
  L_1^z = L_1^{z}(x) := a a' (x-z) \partial_x^2 + b \partial_x,
\end{equation}
and, $L_2^{z} = L_{2,x}^{z} + t L_{2,t}^{z}$, where
\begin{equation} \label{eq.Taylor.L2}
    L_{2,x}^{z} := ( \aonesq + a \atwo )(x-z)^2 \partial_x^2/2 +
   b' (x-z) \partial_x + c, \quad L_{2,t}^{z} :=
    a\dot{a}\partial_x^2.
\end{equation}
Hence
\begin{equation*}
  L^{s,z}(t,x) = L_0^{z} + s L_1^{z}(x) + s^2 \big( L_{2, x}^{z}(x) +
  t L_{2,t}^{z} \big) + s^3 L_3^{s,z}(t,x),
\end{equation*}
where $L_3^{s,z}(t,x)$ is the remainder term.

\begin{remark}\label{remark.1}
Each $L_k^{z}$ in \eqref{eq.Taylor} has polynomial coefficients of
order $k$ in $(x-z)$ and of order $\le k/2$ in $t$. In particular,
$L_0^{z}$ is a constant coefficient operator, for which the Green's
function is computed explicitly in \eqref{eq.e^L_0}.  Thus, in order
for the expansion to capture the time dependence of the coefficients,
the coefficient must be expanded at least to second order in
$s$. Time-dependent corrections will therefore appear only at order
$s^2 =t$ in the expansion of $\cG_t^L(x,y)$.
\end{remark}

Let $\cG_t^L$ be the Green function of the parabolic problem
\eqref{eq.U.tau}, that is, the solution is given by \ $U(t,x) =
\int \cG_t(x,y) h(y) dy =: (\cG_t^L\, h)(t,x)$.

We begin the approximation scheme for $\cG^L_t$ by decomposing $L$
into a constant-coefficient, second-order operator $L_0$, for which we
can explicitly compute the solution operator, and a remainder:
\begin{align}
      L(t) = L_0 + V(t)
\end{align}
where $V(t)$ is a time-dependent, variable coefficient, second order
operator.

By Duhamel's principle we then have
\begin{align} \label{eq.S.1}
  \cG_t^L = e^{t L_0} + \int_0^t e^{(t-\t_1)L_0} V(\t_1)
  \cG_{t_1} ^Ld \t_1.
\end{align}
Repeated applications of Duhamel's formula leads to a recursive
representation of $\cG_t^L$ as a time-ordered expansion:
\begin{multline}\label{eq.time.ordered.formula}
  \cG_t^L = e^{t L_0} + \int_0^{t} e^{(t-\t_1)L_0} V(\t_1  )
  e^{\t_1 L_0} d \t_1 \\
  + \int_0^{t} \int_0^{\t_1} e^{(t -\t_1)L_0} V (\t_1 )e^{(\t_1
    -\t_2)L_0} V(\t_2 ) e^{\t_2 L_0} d\Bar{\t} + \cdots \\
  + \int_0^{t} \int_0^{\t_1} \cdots \int_0^{\t_{d-1}} e^{(t -
    \t_1)L_0} V(\t_1 ) e^{(\t_1 -\t_2)L_0} V(\t_2 ) \cdots
  V(\t_d ) e^{\t_{d} L_0} d\Bar{\tau} \\
  + \int_0^{t} \int_0^{\t_1} \cdots \int_0^{\t_{d+1}} e^{(t
    -\t_1)L_0} V(\t_1 ) e^{(\t_1 - \t_2)L_0} V(\t_2 ) \cdots
  V(\t_{d+1} ) \cG_{\t_{d+1}}^L d\Bar{\tau}
\end{multline}
where, for notational convenience, we have set $d \t_{k} \cdots d \t_2
d\t_1 = d\Bar{\tau}$.  This expansion can be rigorously justified, at
least in the case when $L$ uniformly strongly elliptic and all the
coefficients of $L$ and their derivatives are bounded. See
\cite{CCMN,CMN} for details.  In the limit $d \to \infty$, it yields
an asymptotic time-ordered series, also called a {\em Dyson series},
for the Green's function. The integer $d$ stands for the {\em
  iteration level} in the time-ordered expansion, which at this point
is distinct from the order $n$ of the Taylor expansion of the operator
$L$. For consistency we need $n\geq d$ \cite{CCMN}. {\em We set from
  now on $d=n=2$.}

A similar formula holds for the Green's function $\cG^{L^{s,z}}_t$ of
the solution operator for the rescaled problem \eqref{eq.U.tau.s}.  We
recall that it is enough to compute an approximate Green's function at
$t=1$ for the rescaled problem by \eqref{eq.correspondence}.  We now
choose the operator $L_0$ to be exactly the zeroth-order Taylor
expansion of $L^{s,z}$, given in \eqref{eq.Taylor}. Then:
\begin{equation*}
      V^{s,z}(t) : = L^{s,z}(t)- L_0 =
      s L_1^{z}(x) + s^2 L_2^{z}(t,x) + s^3
      L_3^{s,z}(t,x).
\end{equation*}
and using \eqref{eq.time.ordered.formula} with $d=n=2$ and $t =1$
yields
\begin{align}\label{eq.general.approx}
  \cG_1^{L^{s,z}} = e^{L_0^z} + s \cI_1^z + s^2 \big( \cI_{1,1}^z +
  \cI_{2,x}^z + \cI_{2,t}^z \big)+ \cR^{s,z},
\end{align}
where
\begin{align}\label{eq.I}
\begin{split}
  & \cI_1^{z} = \int_0^{1} e^{(1-\t_1)L_0^z} L_1^z e^{\t_1 L_0^z} d \t_1,\\
  & \cI_{1,1}^{z} = \int_0^{1} \int_0^{\t_1} e^{(1 -\t_1)L_0^z} L_1^z
  e^{(\t_1 -\t_2)L_0^z} L_1^z e^{\t_2 L_0^z} d\t_2 d\t_1, \\
  & \cI_{2, x}^{z} = \int_0^{1} e^{(1-\t_1)L_0^z} L_{2,x}^z e^{\t_1 L_0^z}
  d \t_1, \\
  & \cI_{2,t}^{z} = \int_0^{1} e^{(1-\t_1)L_0^z} \tau_1 L_{2,\tau}^z
  e^{\t_1 L_0^z} d \t_1.
\end{split}
\end{align}
Even though we set $t=1$, we still keep the $t$ dependence explicit in
$\cI_{2,t}^z$ to emphasize this term comes from Taylor expansion in
$t$.  The term $\cR^{s,z}$ in \eqref{eq.general.approx} contains all
the higher order terms and will be included in the remainder.

The approximation for the Green's function $\cG^L_t$
of the original problem \eqref{eq.U} is now obtained as follows. Let
\begin{equation} \label{eq.Tkerneldef}
  T^{s,z}(x,y) = :  G_{0}(x,y;z) + s\, G_{1}(x,y;z) + s^2\,
   G_{2}(x,y;z)
\end{equation}
be the distribution kernel of the operator $e^{L_0^z} + s\, \cI_1^z +
s^2 \big( \cI_{1,1}^z + \cI_{2,x}^z + \cI_{2,t}^z \big)$. The desired
second order approximation is then given by
\begin{equation} \label{eq.kernelTimedep}
  \cG_t^{[2]}(x,y) = t^{-1/2} T^{\sqrt{t},z}(z + (x-z)/\sqrt{t}, z +
  (y-z)/\sqrt{t}),
\end{equation}
where $z=z(x,y)$ is an admissible function. In particular, the kernels
$G^{[n]}_t$ appearing in \eqref{eq.G.expansion} are given by
\begin{equation*}
  G^{[n]}_t (x,y) : = t^{-1/2} G_{n}(z + (x-z)/\sqrt{t}, z +
  (y-z)/\sqrt{t}; z(x,y)).
\end{equation*}

We thus need to compute the distribution kernels of the operators
$e^{L_0^z}$, $\cI_1^z$, $\cI_{1,1}^z$, $\cI_{2,x}^z$, $\cI_{2,t}^z$.
In order to do so, we exploit the semigroup property of $e^{t\,L_0}$
to carry out explicitly the time integration in \eqref{eq.I}.  Before
we proceed, we introduce some useful notation.

By $[T_1,T_2]:=T_1 T_2 - T_2 T_1 = - [T_2, T_1]$ we shall denote the
commutator of two operators $T_1$ and $T_2$.  Two operators $T_1,T_2$
commute if $[T_1,T_2] = 0.$ For any two operators $T_1$ and $T_2$, we
define $ad_{T_1}(T_2)$ by $ad_{T_1}(T_2):=[T_1,T_2],$ and, for any
integer $j$, we define $ad^j_{T_1}(T_2)$ recursively by
\begin{equation*}
    ad^j_{T_1}(T_2):=ad_{T_1}(ad^{j-1}_{T_1}(T_2)).
\end{equation*}

We next recall a Baker-Campbell-Hausdorff-type identity proved and
used in this setting in \cite{CCMN} (note that the operators $T_i$ are
unbounded).  Namely, for any $\theta \in (0,1)$ and differential
operator $Q=Q(x,\partial)$ with {\em polynomials coefficients} in $x$,
we have
\begin{equation}
    e^{\theta L_0^{z}}Q = P_{ad}(Q,\theta,x,z,\partial)e^{\theta
      L_0^{z}},
\end{equation}
where $P_{ad}(Q,\theta,x,z,\partial)$ is a differential operator with
polynomial coefficients in $x$ given by
\begin{align}
  P_{ad}(Q, \theta,x,z,\partial) = Q +
  \sum_{i=1}^\infty\frac{\theta^i}{i!}ad_{L_0^{z}}^i(Q).
\end{align}
In proving this formula, we use the fact that the series is actually a
finite sum, as we show below. In particular, $P$ can be {\em
  explicitly} computed.

A simple calculation gives the following lemma.

\begin{lemma}\label{lemma.f.comm}
Let $L_m$ be a second-order differential operator with
polynomial coefficients of degree at most $m$. Then $ad_{L_0^{z}}^j(L_m) = 0$
for $j>m$. In particular, we have
$[L_{0}^{z}, L_{2, t}^{z}] = 0$, $ad_{L_0^{z}}^2(L_1^z) = 0$,
and $ad_{L_0^{z}}^3(L_{2, x}^z) = 0$.
\end{lemma}

\begin{proof} The proof is a simple calculation.
\end{proof}

We now proceed to compute the integrals in \eqref{eq.I}
\begin{eqnarray*}
\cI_1&=&\int_0^1 e^{(1-\t_1) L_0^{z}}L_1^{z}e^{\t_1L_0^{z}}d\t_1
=  \int_0^1 (L_1^{z}+(1-\t_1) [L_0^{z},L_1^{z}])e^{L_0^{z}}d\t_1\\
& = & (L_1^{z}+\frac{1}{2}[L_0^{z},L_1^{z}])e^{L_0^{z}},\\
\end{eqnarray*}
\begin{eqnarray*}
    \cI_{1,1}&=&\int_0^1 \int_0^{\t_1} e^{(1- \tau_1)
      L_0^{z}}L_1^{z}e^{(\tau_1-\tau_2)L_0^{z}}L_1^{z}e^{\t_2
      L_0^{z}}d\t_2 d \t_1,\\
    &=&\int_0^1
    \int_0^{\t_1}(L_1^{z}+(1-\tau_1)[L_0^{z},L_1^{z}])(L_1^{z}+
      (1-\t_2) [L_0^{z},L_1^{z}])e^{L_0^{z}}d\tau_2 d \t_1, \\
        &=& ( \frac{1}{2} (L_1^{z})^2 +
        \frac{1}{3} L_1^{z}[L_0^{z},L_1^{z}] +
        \frac{1}{6}[L_0^{z},L_1^{z}]L_1^{z} +
        \frac{1}{8}[L_0^{z},L_1^{z}]^2)e^{L_0^{z}},\\
\end{eqnarray*}
\begin{eqnarray*}
\cI_{2,x} & = &
\int_0^1 e^{(1-\t_1) L_0^{z}}L_{2,x}^{z} (\t_1)e^{\t_1L_0^{z}}d\t_1 \\ & =
&\int_0^1 (L_{2,x}^{z} +(1-\t_1)
    [L_0^{z},L_{2,x}^{z}]+\frac{(1-\t_1)^2}{2}[L_0^{z},[L_0^{z},L_{2,x}^{z}]
     ])e^{L_0^{z}}d\t
    \\ & = & \Big( L_{2,x}^{z} + \frac{1}{2}[L_0^{z},L_{2,x}^{z}] +
    \frac{1}{6}[L_0^{z},[L_0^{z},L_{2,x}^{z}]]  \Big) e^{L_0^{z}}, \\
\end{eqnarray*}
\begin{eqnarray*}
    \cI_{2,t} & = & \int_0^1 e^{(1-\t_1) L_0^{z}}\tau_1
    L_{2,\tau}^{z} e^{\t_1L_0^{z}}d\t_1 = \int_0^1 \tau_1L_{2,
      \tau}^{z}e^{L_0^{z}}d\tau_1 = \frac{1}{2}L_{2,t}^{z} e^{L_0^{z}}.
\end{eqnarray*}
Hence \eqref{eq.general.approx} becomes
\begin{equation}
    e^{L^{s,z}}=  \big(1 + sQ_1 + s^2 Q_2\big)e^{L_0^{z}} + \cR^{s,z},
\end{equation}
where
\begin{equation}
\begin{gathered}\label{eq.defQ}
  Q_1 = L_1^{z} + \frac{1}{2} [L_0^{z}, L_1^{z}], \\
  Q_2  = \frac{1}{2} (L_1^{z})^2 +  \frac{1}{3} L_1^{z}[L_0^{z}, L_1^{z}] +
  \frac{1}{6} [L_0^{z}, L_1^{z}] L_1^{z} + \frac{1}{8} [L_0^{z},
 L_1^{z}]^2, \\
  + L_{2,x}^{z} + \frac{1}{2} [L_0^{z}, L_{2, x}^{z}] + \frac{1}{6}
  [L_0^{z}, [L_0^{z}, L_{2, x}^{z}]] + \frac{1}{2} L_{2,t}^{z},
\end{gathered}
\end{equation}
and $\cR^{s,z}$ is again the error term as in \eqref{eq.general.approx}.
Therefore, we only need to compute the commutators in the above
formula to get the second-order approximation of $\cG_t^{L^{s,z}}$.

We recall that we agreed that all functions in the commutator formulas
below are evaluated at $(0,z)$. Hence $a = a(0,z)$, $a' = a'(0,z)$,
and so on. We have
\begin{equation}\label{eq.formula1}
 [L_0^{z},L_1^{z}] = a(0,z)^3 a ' (0,z) \partial_x^3 = a^3
 a'\partial_x^3, \quad [L_0^{z},L_1^{z}]^2 = a^6 \aonesq
 \partial_x^6,
\end{equation}
and hence
\begin{align}
\begin{split}
\label{eq.formula2}
 L_1^{z} [L_0^{z},L_1^{z}] & = a^4 \aonesq (x-z) \partial_x^5 + b a^3
 a' \partial_x^4, \\
 [L_0^{z},L_1^{z}]L_1^{z} & = a^4 \aonesq(x-z)\partial_x^5 + \left( b +
 3 aa' \right) a^3 a' \partial_x^4.
\end{split}
\end{align}

To compute the other commutators, we need the following lemma, which
can be proved by induction using that $[AB, C] = A[B, C] + [A, C]B$.
In particular, \ $[\partial_x^2, x-z] = 2 \partial_x$ and
$[\partial_x^2, (x-z)^2] = 2 + 4(x-z)\partial_x$.

\begin{lemma} For $i, j \ge 1$ integers we have
\begin{equation*}
  \partial_x^i (x-z)^ju(x) = \sum_{k=0}^{\min\{i,j\}} k!
  {i \choose k} {j \choose k} (x-z)^{j-k} \partial_x^{j-k} u(x).
\end{equation*}
\end{lemma}

We therefore have:
\begin{align}
\begin{split}
 [L_0^{z},L_{2, x }^{z}] & = a^2 \left( \aonesq + a \atwo
 \right)(x-z) \partial_x^3, \\
 & + a^2 \left( b' + \aonesq/2 + a \atwo/2 \right) \partial_x^2,
\end{split}
\end{align}
and hence
\begin{equation}
 [L_0^{z},[L_0^{z},L_{2, x}^{z}]] = a^4 (\aonesq + a \atwo ),
 \partial_x^4,
\end{equation}
so that finally
\begin{equation}
  (L_1^{z})^{2} = (a a'(x-z))^2 \partial_x^4 +  2( a^2 \aonesq
  + a a' b )  (x-z) \partial_x^3 + \big( aa' b +
  b^2 \big) \partial_x^2.
\end{equation}

It follows that the approximation kernel of $\cG^{L^{s,z}}_1$ is given
by the applications of a differential operator with polynomial
coefficients to the Green's function of $e^{t\,L_0}$. If $\phi$ is a
smooth function, we denote by \  $C_\phi$ \ the convolution operator with
$\phi$, then \ $C_\phi f(x) := \phi*f(x) = \int \phi(x-y) f(y)dy$,
which shows that the distribution kernel of \ $C_{\phi}$ is $C_{\phi}(x,
y) = \phi(x-y)$. It is immediate to check that
\begin{equation}\label{eq.convolution}
  \partial_x C_\phi = C_{\partial_x \phi},
\end{equation}
while \ $C_\phi \partial_x = - C_{\partial_x \phi}$. By Remark
\ref{remark.1}, the distribution kernel of $e^{L_0^z}$ is given by
\begin{equation} \label{eq.e^L_0}
  e^{L_0^z} (x, y)= \frac{1}{\sqrt{2 \pi a^2 }} \exp( -
    \frac{|x-y|^2}{2 a^2} ), \quad a = a(0,z),
\end{equation}
and hence $e^{L_0^z}$ is a convolution operator.

Then, by \eqref{eq.convolution} \ $\partial_x^k e^{L_0^z}(x,y)=
H_k(\Theta)e^{L_0^z}(x,y)$, where \ $\Theta = \frac{x-y}{a^2}$ and
\ $H_k$ are the (rescaled Hermite) polynomials satisfying $H_0
= 1$ and $H_{k+1}(\Theta) = -\Theta H_{k}(\Theta) + H_{k}'
(\Theta)/a^2$. The polynomials $H_k$ are easily computed by induction
as:
\begin{equation}\label{eq.derivatives}
\begin{gathered}
  H_1(\Theta) = - \Theta, \quad
  H_2(\Theta) = \Theta^2 - \frac{1}{a^2}, \quad
  H_3(\Theta) = - \Theta^3 +\frac{3\Theta}{a^2},\\
  H_4(\Theta) = \Theta^4 - \frac{6\Theta^2}{a^2}
   + \frac{3}{a^4}, \quad
  H_5(\Theta) = - \Theta^5 + \frac{10}{a^2} \Theta^3 -
  \frac{15\Theta}{a^4} ,\\
  H_6(\Theta) = \Theta^6 -\frac{15}{a^2}
  \Theta^4 + \frac{45}{a^4} \Theta^2 - \frac{15}{a^6}.
\end{gathered}
\end{equation}
Using \eqref{eq.derivatives} we have
\begin{align}\label{eq.G0again}
  G_{0}(x,y;z) = e^{L_0^z} =  \frac{1}{\sqrt{2 \pi  a^2
    }} \exp( - \frac{|x-y|^2}{2 a^2} ),
\end{align}
and
\begin{equation} \label{eq.G1}
\begin{split}
  & G_{1}(x,y;z) = \Big( (L_1 + \frac{1}{2} [L_0,L_1]) e^{L_0}  \Big) (x,y) \\
  & = \left( b \partial_x + a a' (x-z) \partial_x^2 +
  \frac{1}{2}a^3 a ' \partial_x^3 \right) e^{L_0^z}(x,y)\\
  &  = \left( bH_1(\Theta) + a a' (x-z) H_2(\Theta) +
  \frac{1}{2}a^3 a ' H_3(\Theta) \right) e^{L_0^z} \\
  & = \frac{1}{\sqrt{2 \pi a^2 }} e^{ - \frac{|x-y|^2}{2 a^2} } \left[
    \left( \frac{3 a a ' - 2 b}{2 a^2} \right)(x-y) -
    \frac{a'}{2a^3}(x-y)^3 \right. \\ & \left. + (x-z) \left(
    \frac{(x-y)^2 - a^2}{ a^3} \right) \right].
\end{split}
\end{equation}

We now carry out a similar calculation for the next (and last) term of
our asymptotic expansion, namely
\begin{equation}
  G_{2}(x,y;z) =  (Q_2 e^{L_0^z})(x,y),
\end{equation}
with $Q_2$ given by Equation \eqref{eq.defQ}.  We finally have
\begin{align}\label{eq.G2}
\begin{split}
  & G_{2}(x,y;z) = \left ( \frac{1}{2} L_{2, \tau}^{z}+ L_{2,x}^{z} +
  \frac{1}{2} [L_0^{z}, L_{2,x}^{z}] + \frac{1}{6} [L_0^{z}, [L_0^{z},
      L_{2,x}^{z}]] \right. \\
  & \left.  + \frac{1}{2}L_1^{z,2} + \frac{1}{3} L_1^{z} [L_0^{z},
    L_1^{z}] + \frac{1}{6} [L_0^{z}, L_1^{z}] L_1^{z}+\frac{1}{8}
       [L_0^{z}, L_1^{z}]^2 \right ) e^{L_0^z}\\ &=\left
       (P_0+\sum_{i=1}^6 P_iH_i (\Theta)\right )e^{L_0^z} (x,y)\, .
\end{split}
\end{align}
where $P_j$ are polynomials in $x-z$ and $x-y$ with coefficients given
in terms of the values of the functions $a$, $b$, and $c$, and their
derivatives, all evaluated at $z = z(x,y)$, as follows
\begin{equation}
\begin{split}
&P_0=c, \ \ P_1=b'(x-z),\\
& P_2=\frac{1}{2} \Big [ \frac{1}{2} a^3 \atwo + a^2 b' + a^2
    \aonesq/2 + b^2 +   \aonesq  (x-z)^2 \\
& \qquad \qquad  + a \left ( ba' + \dot{a} + \atwo(x-z)^2 \right ) \Big ],\\
&P_3=a(x-z)( a'b + \frac{1}{2} a^2 \atwo + \frac{3}{2}
  a \aonesq),\\
&P_4=\frac{a^2}{3} \Big[ \frac{1}{2} a^3 \atwo + 2 a^2 \aonesq + a
         \frac{3}{2} a' b + \frac{3}{2} \aonesq(x-z)^2 \Big ],\\
&P_5=\frac{1}{2} a^4\aonesq(x-z), \ \ P_6=\frac{1}{8} a^6 \aonesq.
\end{split}
\end{equation}

In particular, we obtain the following explicit formula.

\begin{example}\label{ex.CEV1}
For the CEV model given by Equation \eqref{eq.LCEV}, we have $a =
\sigma z^{\alpha}$, $a' = \alpha \sigma z^{\alpha-1}$, $b =
rz^{\alpha}$, $z = z(x,y)$, and hence,
\begin{multline*}
  G^{CEV}_{1}(x,y;z)  = \frac{1}{\s z^\a \sqrt{2 \pi}}
\left[ \left( \frac{3 \alpha
    \sigma^2 z^{2(\alpha - 1)} - 2 r}{2 \sigma} \right) \left(
  \frac{x-y}{\sigma z^\a} \right) \right.
  \\  - \frac{\alpha
    \sigma z^{\alpha - 1}}{2} \left( \frac{x-y}{\sigma z^\alpha}
  \right)^3 +  \left.   \left( \frac{x-z}{\sigma z^\alpha} \right)
  \left( \left( \frac{x-y}{\sigma z^\a} \right)^2 - 1 \right)
  \right]
e^{ - \frac{|x-y|^2}{2 \s^2 z^{2\a}} }.
\end{multline*}
\end{example}

Let us introduce the {\em time dependent Black-Scholes-Merton model}
to correspond to the operator
\begin{align}\label{eq.tLBSM}
  L := \frac{1}{2} \sigma(t)^2 x^2 \D_x^2 + r(t) x\D_x - r(t),
\end{align}
Thus the difference between the usual Black-Scholes-Merton model
\eqref{eq.LBSM} and the time dependent Black-Scholes-Merton model
\eqref{eq.tLBSM} is that in the latter we allow $\sigma$ and $r$ to
depend on time.  Then the asymptotic formula for the time dependent
Black-Scholes-Merton model is obtained by setting $\alpha = 1$ in the
Example \ref{ex.CEV1}, since that formula does not contain time
derivatives of the coefficients.

At this stage, we can allow the basepoint $z$ to vary with $x$ and
$y$. In Section \ref{sec.basepoints} below we compute the expansion
for the basepoint $z(x,y)=x$ and compare it in Section
\ref{sec.error}.  Different choices of basepoints $z$ may lead to more
accurate and stable approximations. In future work, we plan to study
how to optimize the choice of $z$.

\begin{definition} \label{def.admissible}
 We call a function $z=z(x,y)$ {\em admissible}  if  $z(x,x)=x$ and all
derivatives of $z$ are bounded.
\end{definition}

In \cite{CCMN,CMN}, we rigorously prove error bounds for the remainder
term in \eqref{eq.U.expansion} in Sobolev spaces under the assumption
that $z$ be admissible (and all the coefficients of $L$, together with
their derivatives, be bounded functions, and $L$ be uniformly
  strongly elliptic).  The function $z=z(x,y)$ can be thus quite
general.

\subsection{Kernel expansions at $z=x$}\label{sec.basepoints}
The choice $z = x$ yields a simplified expression for the
approximation, since certain terms disappear, and the approximation
yields the price of a European call option in closed form. In fact,
the convolution with the approximate Green's function can be evaluated
exactly and the price of a European call option given in closed form.
In particular there is no need for numerical quadrature in evaluating
the integrals, thus improving the speed of our calculations.

\begin{example}
By setting $z = x$ in\eqref{eq.G1} and evaluating all
  coefficient functions at $(0,x)$, we obtain the first-order
correction to the rescaled kernel $\cG^{L^{s,z}}_1$ in the form:
\begin{equation}\label{first_order_approx}
  G_{1}(x,y; z=x) =\frac{x-y}{\sqrt{2\pi a^2}}
  e^{-\frac{(x-y)^2}{2a^2}} \left(\frac{3a a'- 2b}{2a^2} -
  \frac{a^\prime}{2a^3}(x-y)^2\right ).
\end{equation}
\end{example}

\begin{example}
Similarly, the second-order correction to the rescaled kernel
$\cG^{L^{s,z}}_1$ is obtained in the form:
\begin{align}\label{second order approx}
\begin{split}
  & G_{2}(x,y;z=x)= \frac{1}{\sqrt{2\pi
      a^2}}e^{-\frac{(x-y)^2}{2a^2}}\cdot \left \{ \frac{1}{8} a^6
  a'^2 H_6(\Theta) \right.\\
  & +\frac{a^{3}}{6}\left ( a^2 \atwo +4aa'^2+3ba' \right )
  H_4(\Theta)\\
  & \left . +  \frac{1}{4}\left (a^3\atwo +2a^2b'+2a\dot{a}+2aa'b +
  a'^2a^2+2b^2 \right ) H_2(\Theta)  +c \right \},
\end{split}
\end{align}
where $\Theta = \frac{x-y}{a(0,x)^2}$, and $H_6, H_4, H_2$ are given
by \eqref{eq.derivatives}.
\end{example}

\begin{example}
For the time-dependent Black-Scholes-Merton equation, we have $b(t,x)
= r\,x$, $c(t,x)=-r$ and $a(t,x) = \sigma(t) \, x$ so that
\begin{align}\label{eq.Green.function.BSM.standard}
G^{BSM}_{1}(x,y;z=x)= \frac{x-y}{\sqrt{2\pi} \sigma^2 x^2} e^{-
  \frac{1}{2}\left( \frac{x-y}{\sigma x} \right)^2} \left[
  \frac{3 \s^2 - 2r}{2 \s} - \frac{\s}{2} \left( \frac{x-y}{\s x}
  \right)^2 \right],
\end{align}
where all coefficient functions are calculated at $(0,x)$.
\end{example}

\begin{example}
The second-order correction to the rescaled kernel is given by
\begin{align}\label{eq.Green.function.BSM.standard.2}
\begin{split}
& G^{BSM}_{2}(x,y;z=x) =\frac{1}{\sqrt{2\pi} \s x} e^{-
    \frac{1}{2}X^2} \cdot \left [
    -\frac{\sigma^4+4\sigma\dot{\sigma}(0)+4r^2+4r\sigma^2
    }{8\sigma^2} \right. \\
& \left. +\frac{ 4
      \dot{\s}(0)\s+4r^2-16r\sigma^2+15\sigma^4
    }{8\sigma^2}X^2+\frac{12r-29\sigma^2}{24}X^4
    +\frac{\sigma^2}{8}X^6\right ],
\end{split}
\end{align}
where $X = (x-y)/(\sigma x).$
\end{example}

\begin{example}
For the time-dependent CEV model, $c(t,x) = -r$ and $a(t,x)
= \sigma(t)\, x^\a$, $b(t,x)=r\,x$ with $\sigma(0) = \s$ so that

\begin{align}
\begin{split}
G^{CEV}_{1}(x,y;z=x) = &\frac{x-y}{\sqrt{2 \pi} \sigma^2 x^{2\alpha}}
e^{-\frac{1}{2} \left( \frac{x-y}{\s x^\a} \right)^2} \left[ \frac{3
    \a \s^2x^{\a-1} - 2 rx^{1-\a} }{2 \s} \right. \\ &\left. -
  \frac{\a \s x^{\a -1}}{2}\left( \frac{x-y}{\s x^\a} \right)^2
  \right],
\end{split}
\end{align}
\end{example}
and
\begin{example} For the second order correction we have
\begin{equation}
G^{CEV}_2(x,y;z=x)=\left ( P_2H_2+P_4H_4+P_6H_6  -r\right
)\frac{1}{\sqrt{2\pi} \sigma x^\a}\cdot
  e^{-\frac{(x-y)^2}{2\sigma^2\, x^{2\a}}}
\end{equation}
where $$H_2=\left( \frac{x-y}{\s^2 x^{2\alpha}}\right
)^2-\frac{1}{\s^2 x^{2\alpha}}, H_4=\left( \frac{x-y}{\s^2
  x^{2\alpha}}\right
)^4-\frac{6(x-y)^2}{\s^6x^{6\alpha}}+\frac{3}{\s^4x^{4\alpha}}$$
$$ H_6=\left( \frac{x-y}{\s^2 x^{2\alpha}}\right
)^6-\frac{15}{\s^2x^{2\alpha}}\left( \frac{x-y}{\s^2
  x^{2\alpha}}\right )^4+\frac{45}{\s^4x^{4\alpha}}\left(
\frac{x-y}{\s^2 x^{2\alpha}}\right )^2-\frac{15}{\s^6x^{6\alpha}}$$
and
$$
P_2 = \frac{1}{4}\left ( \s^4\alpha (2\alpha-1)x^{4\alpha-2}+2\s (\s r + \dot{\s}(0)+\s \alpha r)x^{2\alpha} +2r^2x^2 \right )
$$
$$ P_4= \frac{1}{6}\s^4 \alpha
x^{4\alpha}\left(\s^2(\alpha-1)x^{2\alpha-2}+4\s^2 \alpha
x^{2\alpha-2}+3r\right), P_6=\frac{1}{8}\s^8 \alpha^2 x^{8\alpha-2}
$$
\end{example}

Note that in the above two Examples for the CEV model, setting
$\alpha=1$ leads to the corresponding approximation for the BSM
model.

\section{Closed Form Approximate Solutions} \label{sec.closedform}

In this section, we consider European call options. For European put
options similar results can also be obtained, either directly from the
definition or by using put-call parity \cite{Shreve}. In what follows,
we will work with the expansion obtained by setting $z(x,y)=x$ as the
basepoint. In this case, we are able to compute the integrals defining
the approximate option price $U^{[k]}$ from $G^{[k]}$ in closed form,
which bypasses the need for more computationally intensive integration
methods such as numerical quadrature, which are needed for more general basepoints $z$.

By \eqref{first_order_approx} and \eqref{eq.correspondence}, the
first-order approximate Green's function in given by
\begin{equation}
\label{eq.first.order}
   \cG_t^{[1]}(x,y) = \frac{1}{\sqrt{2\pi t}a}e^{-\frac{(x-y)^2}{2a^2
       t}}\left(1 + \frac{3aa'- 2b }{2a^2}(x-y) - \frac{a'}{2a^3
     t}(x-y)^3 \right ).
\end{equation}

Similarly, by \eqref{second order approx} and
\eqref{eq.correspondence} the second-order approximate Green's
function is given by

We recall that we implicitly assume all coefficients are evaluated at
$(0,x)$.

\begin{equation}
\cG_t^{[2]}(x,y)= \frac{\sqrt{t}}{\sqrt{2\pi}
    a}e^{-\frac{(x-y)^2}{2a^2t}} \cdot \left ( P_6 H_6(\Xi)+P_4H_4(\Xi)+P_2 H_2(\Xi)-r \right )
\end{equation}
where $\Xi=\frac{x-y}{a^2\sqrt{t}}$, the functions $H_6,H_4,H_2$ are
given by \eqref{eq.derivatives}, and
\begin{equation}
\begin{split}
P_6=\frac{1}{8}a^6 a'^2, P_4=\frac{a^{3}}{6}\left ( a^2 \atwo
+4aa'^2+3ba' \right )\\
P_2=\frac{1}{4}\left (a^3\atwo +2a^2b'+2a\dot{a}+2aa'b +
  a'^2a^2+2b^2 \right )
\end{split}
\end{equation}
All the coefficient functions are evaluated at $t=0, z=x$.

For European Call options with strike
price $K$, by \eqref{eq.convolution.U}  the $n^{th}$-
order approximated option price is
\begin{equation} \label{option price}
U^{[n]}(t,x)=\int_0^\infty \cG_t^{[n]}(x,y)(y-K)^+dy,
\end{equation}
where we only take $n=1,2$ here.  We recall that $t=T-\mathfrak{t}$ is
the time to expiry $T$ and $\mathfrak{t}$ is real time. So to be more
precise, {\em the $n^{th}$-order option pricing formula for European
  call options with expiry time $T$ is \ $
  U^{[n]}(T-\mathfrak{t},x)$}.

We have already observed that the general form of the approximate
kernel, when $z=x$, is a product of polynomial functions against a
rescaled Gaussian.  Therefore, the integration in \eqref{option price}
above can be carried out in terms of error functions.  Explicitly,
\begin{eqnarray*}
U^{[1]}(t,x)&=&\frac{\sqrt{t}}{2\sqrt{2\pi
}}e^{-\frac{(x-K)^2}{2a^2 t}}\left (
2a-a'(x-K)
\right )\\
&&+\frac{1}{2}\cdot \left ( {\rm erf} \left(\frac{x-K}{\sqrt{2
t}a}\right)+1\right )\left(b\, t+x-K\right),
\end{eqnarray*}
and
\begin{equation}\label{2nd closed form}
\begin{split}
&U^{[2]}(t,x)=U^{[1]}(t,x)+\frac{rt}{2} \left ( erf\left(
\frac{x-K}{a\sqrt{2 t}} \right)+1 \right )(K-x)\\
&+\frac{1}{2\sqrt{2\pi  t}}e^{-\frac{(x-K)^2}{2a^2 t}}\cdot
\left ( \frac{t^2}{6} a^2\atwo
-a(r+a'^2/12)\, t^2
\right.\\
&+( t\, \dot{a}+\atwo  (x-K)^2/3) \, t+\frac{t}{a}
(t\, b^2-a'^2(x-K)^2/6)\\
&\left.+\frac{t\, b a'
(x-K)^2}{a^2}+\frac{a'^2(x-K)^4}{4a^3} \right ),
\end{split}
\end{equation}
Note that in financial applications (\ie in the risk free measure)
$b(t,x)=rx$ and $c(t,x)=-r$.

\begin{example} For the CEV model we have
\begin{align}\label{eq.CEV.approx}
\begin{split}
U_{CEV}^{[1]}(t,x)&=\frac{\sigma x^{\a - 1} \sqrt{t}}{2\sqrt{2\pi
}}e^{-\frac{(x-K)^2}{2  \sigma^{2}t\, x^{2\a}}}\left ( (2 -
\a)x + \alpha K \right )\\ &+\frac{1}{2}\cdot \left (
  {\rm erf} \left(\frac{x-K}{\sqrt{2 t}\sigma x^\a}\right)+1\right
  )\left( (1 + r t)\,x - K \right).
\end{split}
\end{align}
and when $\alpha=1$, it reduces to the first order approximation for
the Black-Scholes-Merton model.
\end{example}

\section{Comparison and Performance of the Method}\label{sec.error}

In this section, we discuss the accuracy and efficiency of our
approximation for the Black-Scholes-Merton and the CEV model.  We
employ the Black-Scholes-Merton model primarily as a didactic
example, given that an exact kernel and option pricing formulas
exists.  For the CEV model, we compare our approximation to other
solution formulas considered a benchmark in the literature, in
particular the Hagan-Woodward scheme \cite{HW}.

What we find in general is a very good agreement of the approximate
pricing formulas we derive in this paper with those available in the
literature, but with significant advantage in the computational
efficiency.  In particular, the agreement is good even for  times
that are not small.  In Section \ref{bootstrap}, we propose a bootstrap
scheme in time to improve the accuracy of our approximation for large
time.

\subsection{Performance of the method}

We start by discussing the Black-Scholes-Merton model, and choose the
parameters $K=15$, $\sigma=0.3$ and $r=0.1$, and plot the exact and
approximate solutions for $0<x<25$. We compare our formula with the
Black Scholes exact solution formula for different times $t$. Figure
\ref{figure BS} gives two different cases, which show that when $t$ is
small the two solutions are in very good agreement with an absolute
error of order $O(10^{-3})$. We notice that even when $t$ is not
small, the error is small. Tables \ref{table 1} and \ref{table 2} give
a analysis of the {\em pointwise} error for the first order
approximation with respect to the exact Black-Scholes formula.

\begin{remark}
Throughtout this section, we {\em fix the basepoint $z=x$}, so that we
have closed-form approximate solution formulas, and we can better
gauge the error introduced by the our method. For more general
basepoints $z$, further error is introduced by the numerical
quadrature used for the integration and the truncation of the pay-off
function $h$ at large $x$ (this error is lower order, however, if $h$
is truncated at $x$ large enough with respect to $K$).
\end{remark}

\begin{figure}
\begin{center}
\includegraphics[width=2.2in]{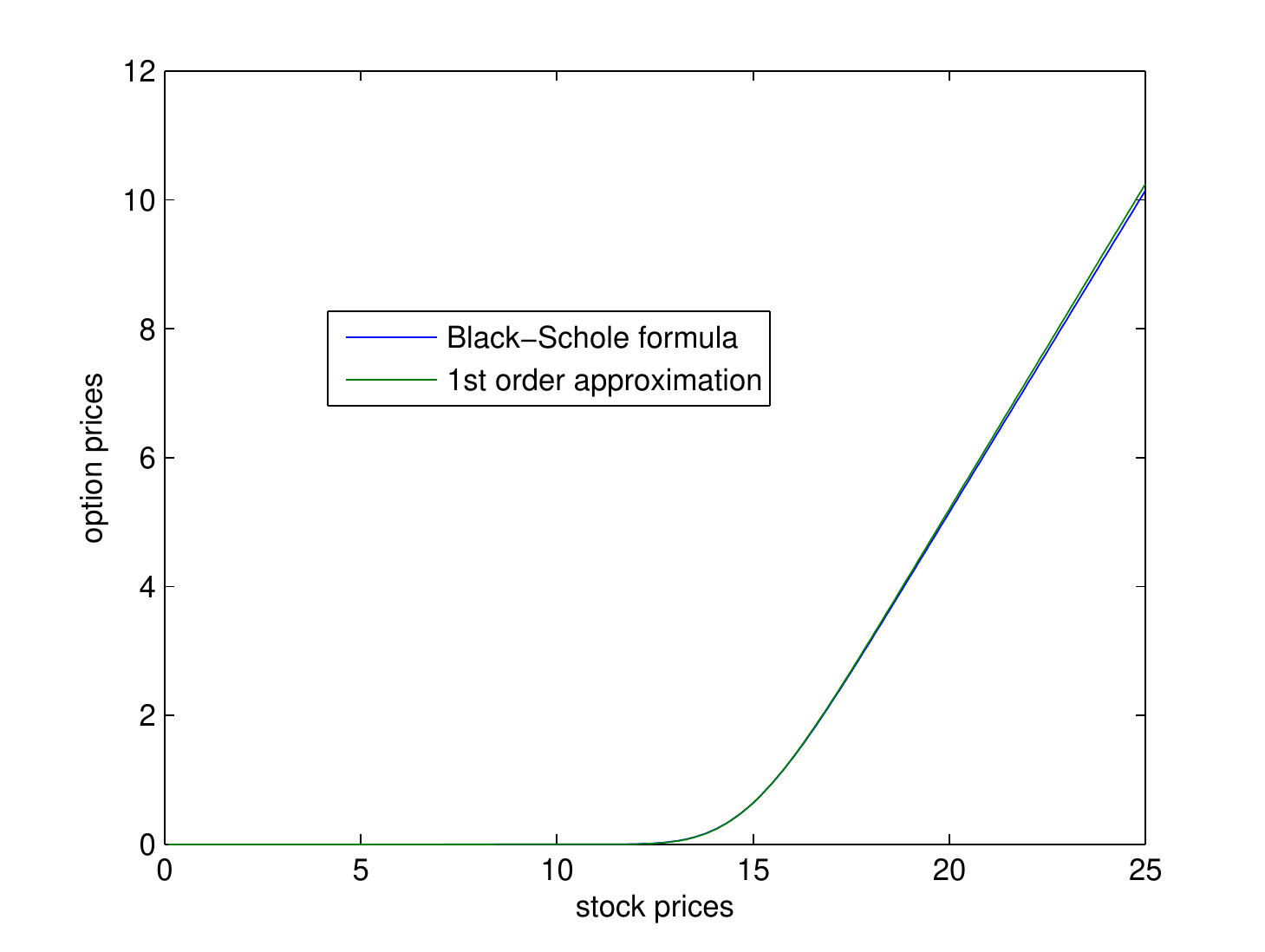}
\includegraphics[width=2.2in]{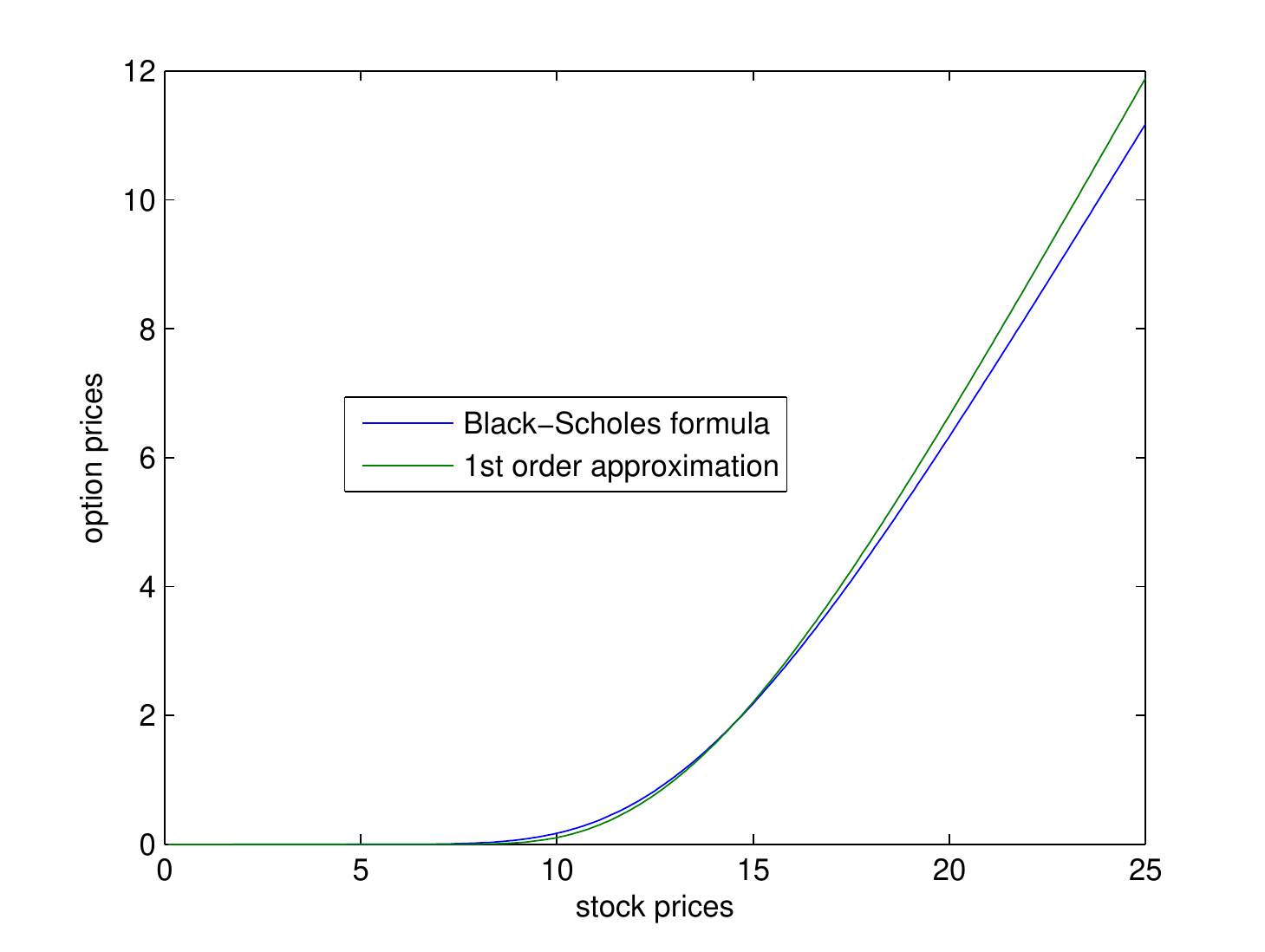}
\end{center}\caption{Comparison of our first order approximation with the Black-Scholes formula.
Parameters: $\,K=15,\,\sigma=0.3,\,r=0.1$. Basepoint $z=x$. The left
picture is for $t=0.1$, and the right one is for $t=0.8$. Note that
the x-axis is scaled by 10, i.e. the label 150 means the stock price
is 15. } \label{figure BS}
\end{figure}

\begin{table}
\begin{tabular}{|l|l|l|l|l|l|l|l|}\hline
 \backslashbox{$t$}{$x$}& 12&13 & 14 & 15 & 16 & 17 & 18 \\\hline
 0.01 &0.0000&0.0000&0.0313&0.3266&0.0387&0.0019&0.0000\\\hline
 0.05 &0.0461&0.3385&0.0179&0.3915&0.0179&0.4068&0.3957\\\hline
 0.1 &0.7&0.7&0.2&0.5&0.2&0.4&1.2\\\hline
 0.2 &2.2&0.3&0.7&0.9&0.7&0.3&1.3\\\hline
 0.5 &1.2&2.1&2.5&2.7&2.7&2.9&1.9\\\hline
\end{tabular}
\medskip
\caption{Error of the first order
approximation for the BSM model, $K=15,\sigma=0.3,r=0$, error
scale= $10^{-3}$.}\label{table 1}
\end{table}

\begin{table}
\begin{tabular}{|l|l|l|l|l|l|l|l|}\hline
 \backslashbox{$t$}{$x$}& 12&13 & 14 & 15 & 16 & 17 & 18 \\\hline
 0.01 &0.0000&0.0000&0.1000&0.0000&0.9000&2.0000&3.0000\\\hline
 0.05 &0.1&0.9&1.4&0.1&3.6&8.7&14.5\\\hline
 0.1 &1.7&3.8&3.3&0.3&7.0&15.9&26.4\\\hline
 0.2 &9.3&10.7&7.1&1.2&14.0&30.2&48.8\\\hline
 0.5 &39.0&31.4&15.1&8.4&39.4&76.0&116.8\\\hline
\end{tabular}
\medskip
\caption{Error of the first order approximation for the BSM model,
  $K=15,\sigma=0.3,r=0.1$, error scale=$10^{-3}$.} \label{table 2}
\end{table}


\begin{remark}
Formula \eqref{eq.BS.approx} shows that the first-order approximation
of the kernel depends linearly on $r\,t$. Therefore, the error grows
more rapidly for $r$ large at comparable times. The same observation
holds for the CEV model. For Black-Scholes, this issue does not arise,
since a change of variables allows to reduce to the case $r=0$ in the
equation.
\end{remark}

Analytic pricing formulas for the CEV model in terms of Bessel
function series have been derived for any value of $\beta$
\cite{cox,EM}. However, sum such series to accurate order can be very
computationally intensive (but see Schroder \cite{Schroder} for
methods to compute the pricing formulas more efficiently).

The numerical tests show our approach yields accurate pricing formulas
that are, however, computationally much simpler. We choose
$\beta=\frac{2}{3},K=15,\sigma=0.3,r=0.1$ for parameters.  Schroder
\cite{Schroder} derived the exact CEV solution when
$\beta=\frac{2}{3}$. Figure \ref{figure CEV} gives the comparison of
our method and the true solution of the CEV model for this value of
$\beta$ for different times.  Again, we plot the two solutions for
$0<x<25$.

\begin{figure}
\begin{center}
\includegraphics[width=2.2in]{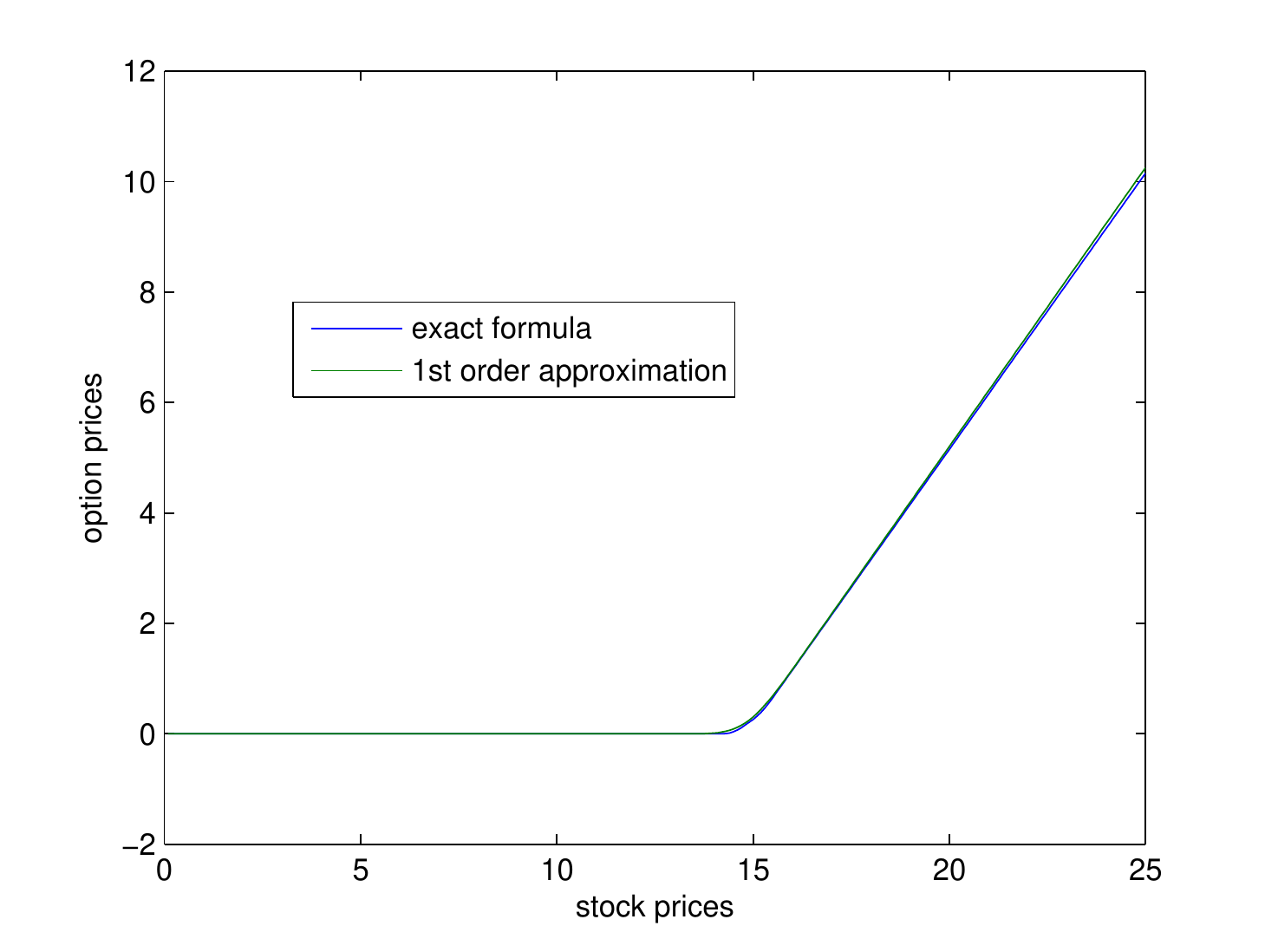}
\includegraphics[width=2.2in]{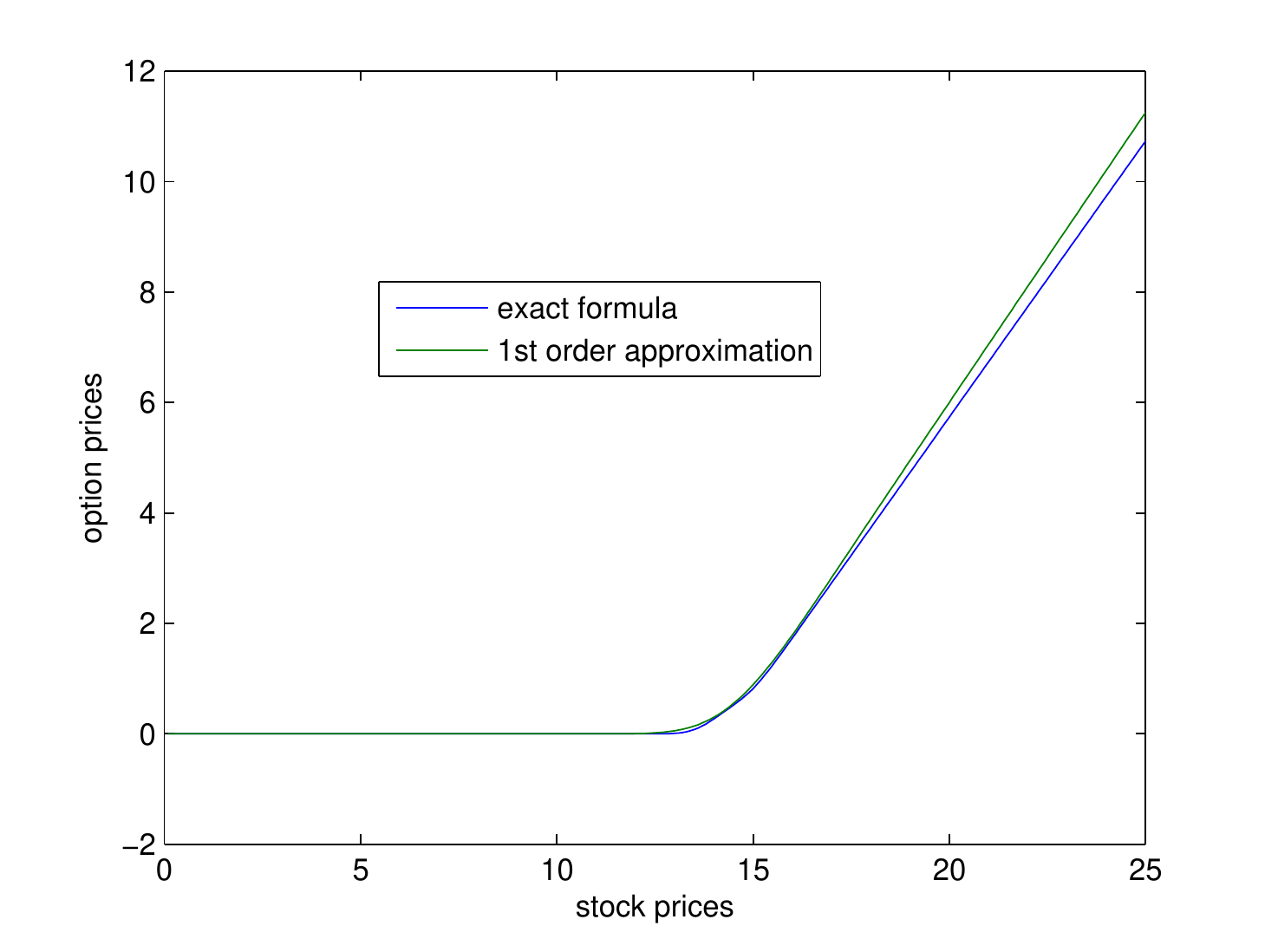}
\end{center}\caption{Comparison of our first order approximation with the exact formula for the CEV model derived in \cite{Schroder}.
Parameters:$\beta=\frac{2}{3},K=15,\sigma=0.3,r=0.1$. Basepoint
$z=x$. The first graph is plotted when $t=0.1$, and the second is
when $t=0.5$. In the graphs the x-axis is scaled by 10, i.e the
lable 150 means the stock price is 15.} \label{figure CEV}
\end{figure}

Hagan and Woodward in \cite{HW} studied more general local volatility
models, for which the stock price under the forward measure follows
the SDE
\begin{equation*}
    dF_t=\gamma(t)A(F_t)dW_t,
\end{equation*}
for some deterministic and suitably smooth functions $\gamma$ and
$A$.  CEV fits into this general model.

Using a singular perturbation technique, Hagan and Woodward obtain a
very accurate formula for the implied volatility
for this model.  In the CEV case, their implied volatility reads
\begin{equation*}
    \sigma_B=\frac{a}{f^{1-\beta}}\left ( 1+\frac{1}{24}(1-\beta)(2+\beta)
    \left(\frac{e^{rT}S_0-K}{f} \right )^2
    +\frac{1}{24}\frac{(1-\beta)^2a^2T}{f^{2(1-\beta)}} \right ),
\end{equation*}
where
\begin{equation*}
  a=\sigma
  \sqrt{\frac{e^{2r(1-\beta)T}-1}{2r(1-\beta)T}},f=\frac{e^{rT}S_0+K}{2}.
\end{equation*}
The approximate pricing formula is then obtained from the
Black-Scholes formula by using $\sigma_B$ as volatility.

When $\beta=\frac{2}{3}$, the CEV formula can be computed exactly
\cite{Schroder}.  In this case, Hagan and Woodward's approximation is
shown by Corielli {\em et al} to be very accurate \cite{CFP}.  We
therefore take this approximation as benchmark for comparison with our
method. In the following numerical comparison, we choose $\beta=2/3$,
$K=20$, $r=0.1$, $\sigma=0.3$ and different times $\tau=0.3, \; 0.5$
We compute the prices on the interval $[0,30]$, and divide it into 300
subintervals. Since the prices near the strike is of most interest for
practitioners, we compare the methods near $K=20$.  Figure \ref{figure
  Hagan} gives the results, from which we see that our approximation
is more accurate than the Hagan-Woodward approximation near the strike
for different times.

We remark that our method can in principle yield arbitrary accuracy in
the small-time limit if more terms in the kernel expansion
\eqref{eq.G.expansion} are included. Furthermore, it allows to derive
approximate solution formulas for even more general models than those
of Hagan and Woodwards.

\begin{figure}
\begin{center}
\includegraphics[width=2.2in]{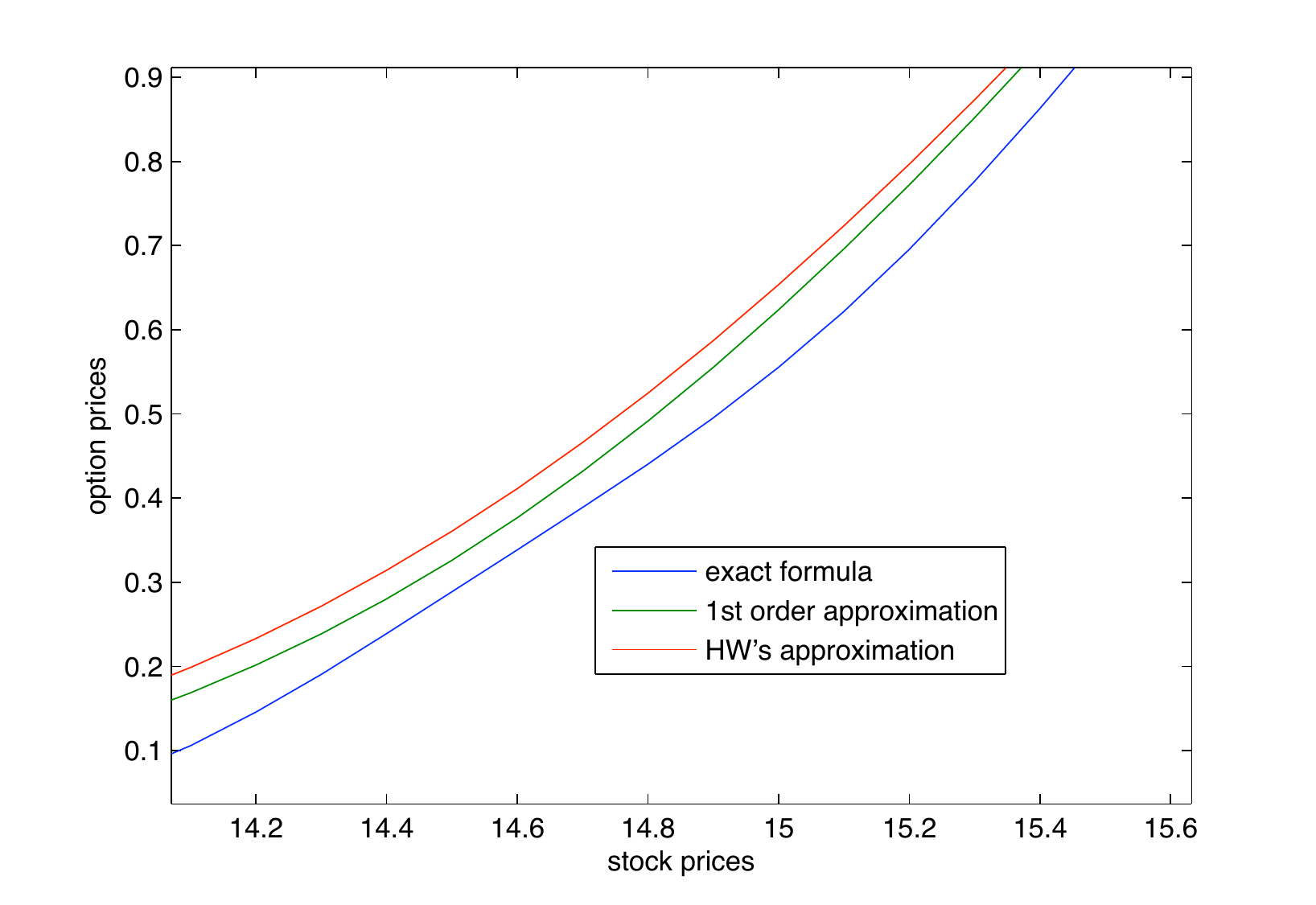}
\includegraphics[width=2.2in]{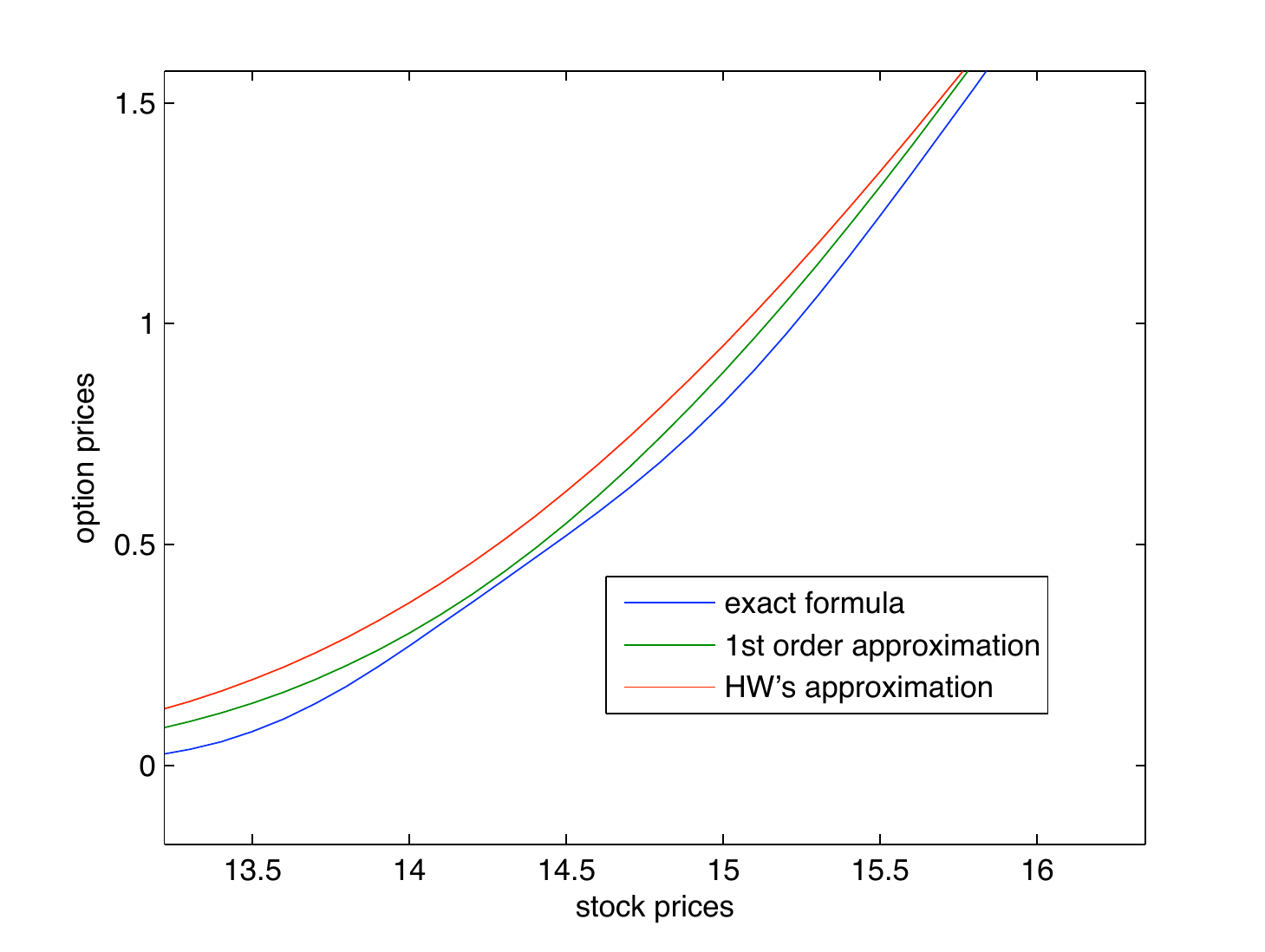}
\end{center}\caption{Comparison of our approximation with Hagan's results for the CEV model near strike.
Parameters: $\beta=2/3$, $K=20$, $r=0.1$, $\sigma=0.3$. Basepoint:
$z=x$. The first graph is plotted when $t=0.3$, and the second is
when $t=0.5$. Note that in the graphs the x-axis is scaled by 10,
i.e. the label 200 denotes that the stock price is 20.}\label{figure
Hagan}
\end{figure}


\subsection{The Greeks}\label{greeks}

In this part, we use the second-order approximate solution to compute
the Greeks of a European call option. The Delta and Gamma of
  a call option, collectively known as the {\em Greeks} of the option,
  at the point $x$ are calculated as
$$
\text{delta}=\frac{u(\t,x+dx)-u(\t,x-dx)}{2dx},
$$
and
$$
\text{gamma}=\frac{u(\t,x+dx)+u(\t,x-dx)-2u(\t,x)}{dx^2}
$$ respectively, where $u(t,x)$ is the option price.  Some methods,
for example the Monte Carlo method, can price options accurately, but
they are not efficient for obtaining good hedging parameters.  We
shall show that our approximations not only give option prices, but
also Greeks accurately.  Again for didactic purposes, we choose the
Black-Scholes-Merton model for which the Greeks can be computed
exactly.

Since we can price options in closed form (by choosing $z=x$), we can
calculate the Greeks in closed form by simply differentiating the
approximate pricing formula. However, again because of the complexity
of these formulas, we will obtain the hedging parameters numerically.

In the numerical experiment, we choose the parameters as follows:
maturity $\t=0.5$, volatility $\sigma=0.5$, strike $K=20$, interest
rate $r=10\%$ In Figure \ref{figure delta}, we plot the difference
between our approximation and the exact solution for Delta when the
stock price varies from 0 to 40.  Figure \ref{figure gamma} does the
same for Gamma. The numerical test shows that the pointwise difference
is very small, of the order of $10^{-3}$ in both cases. More
specifically, the biggest error is around $13\times 10^{-3}$.

\begin{figure}
\begin{center}
\includegraphics[width=2.2in]{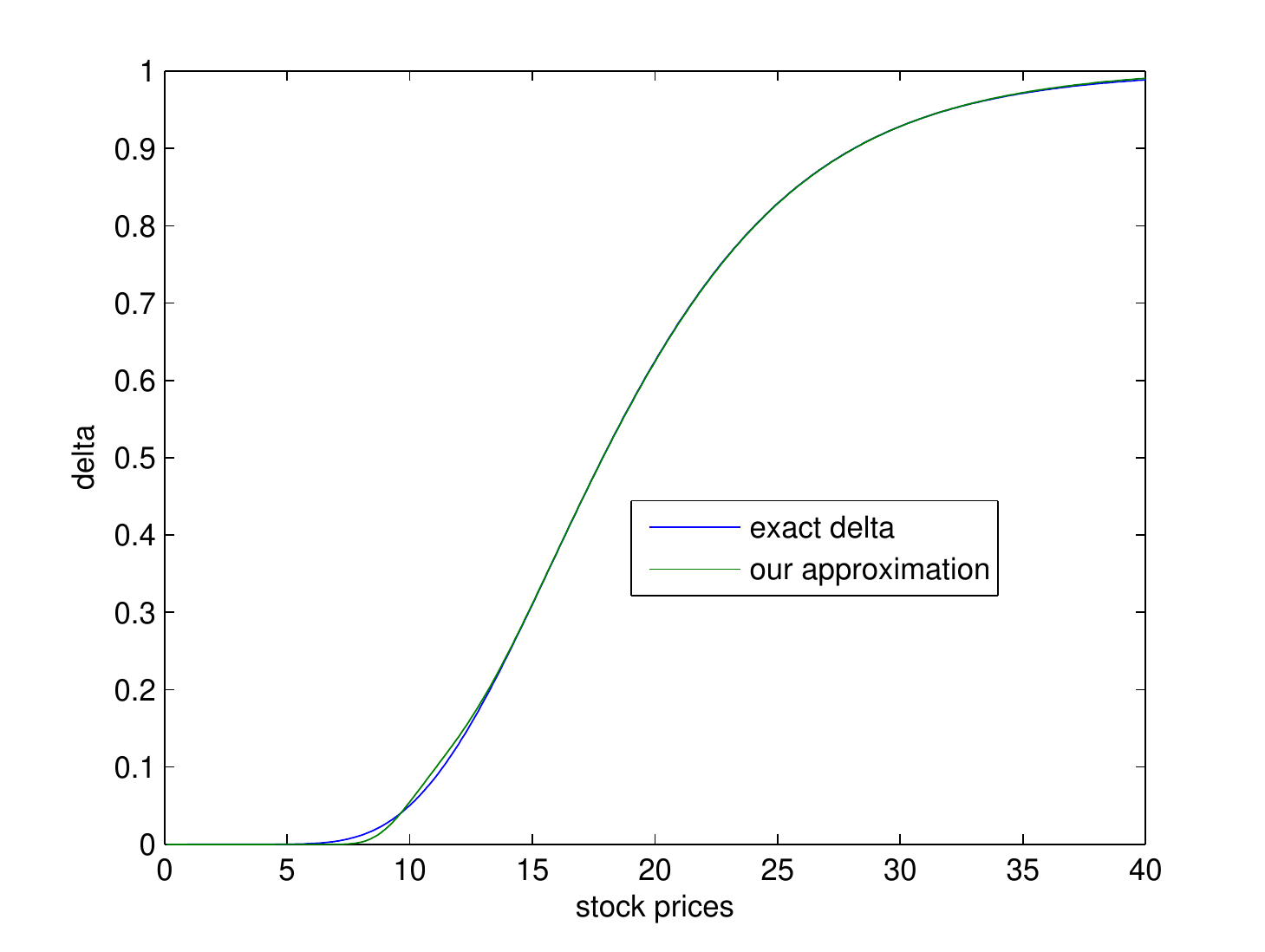}
\includegraphics[width=2.2in]{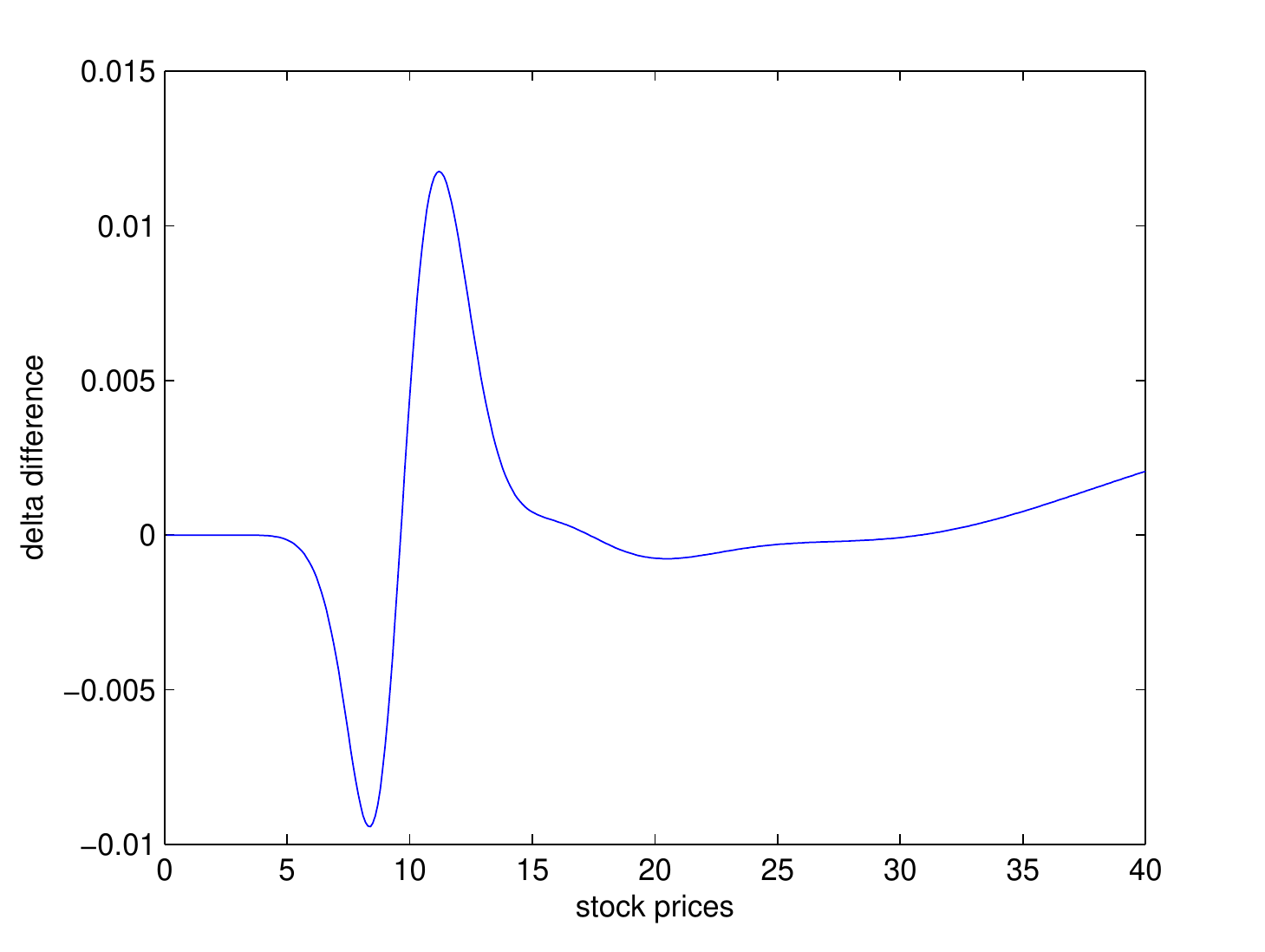}
\end{center}
\caption{Comparison of Delta of our approach and the true values under
  the Black-Scholes model.  Model parameters:
  $t=0.5,K=20,\sigma=0.5,r=10\%$. Basepoint: $z=x$.  The left graph
  plots the delta computed by our method and the true delta. The right
  graph plots their difference. Note that in the second figure the
  scale is $10^{-3}$. The x-axis is scaled by 10, that is, the label
  400 means the stock price is 40.}\label{figure delta}
\end{figure}

\begin{figure}
\begin{center}
\includegraphics[width=2.2in]{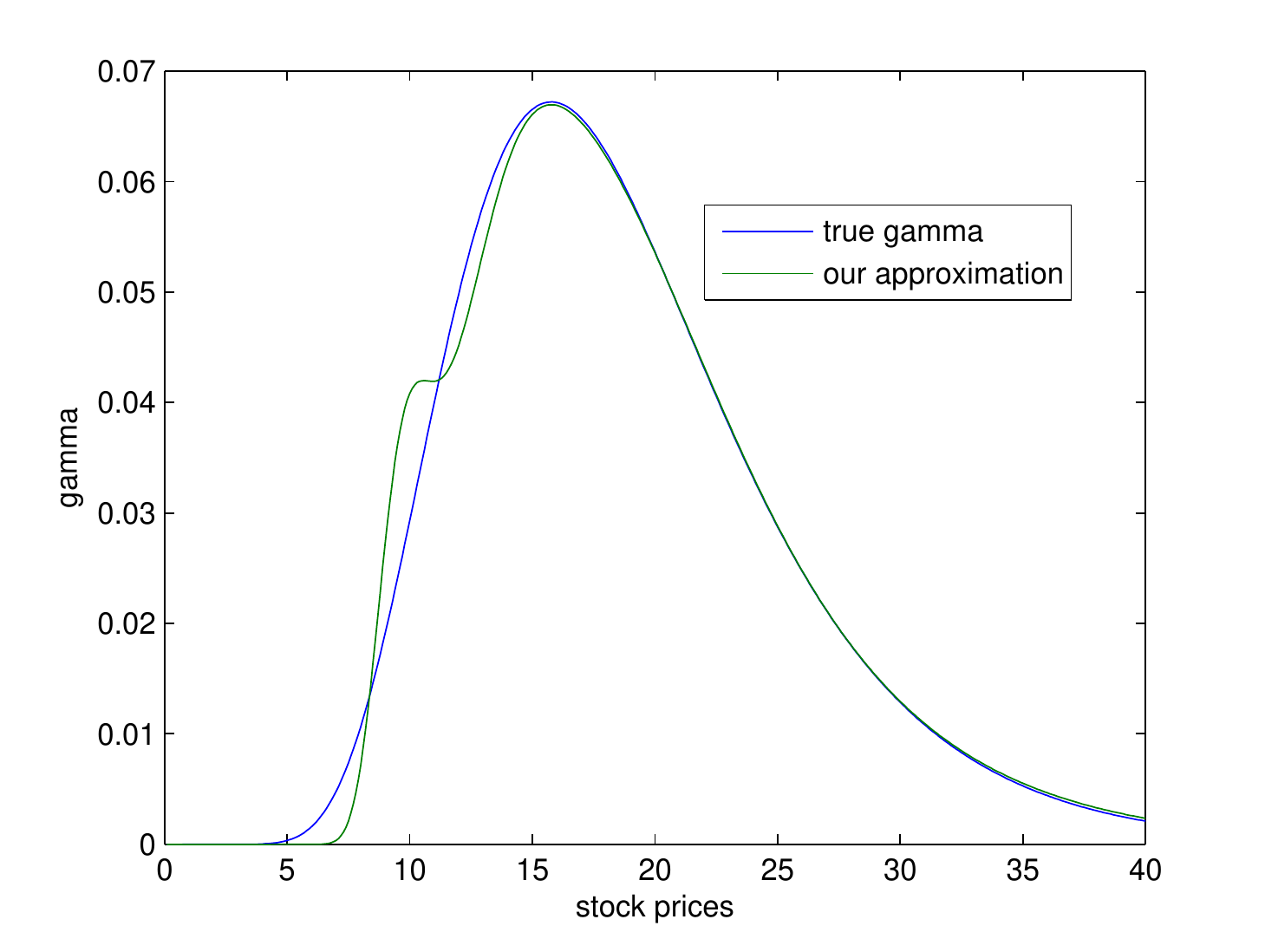}
\includegraphics[width=2.2in]{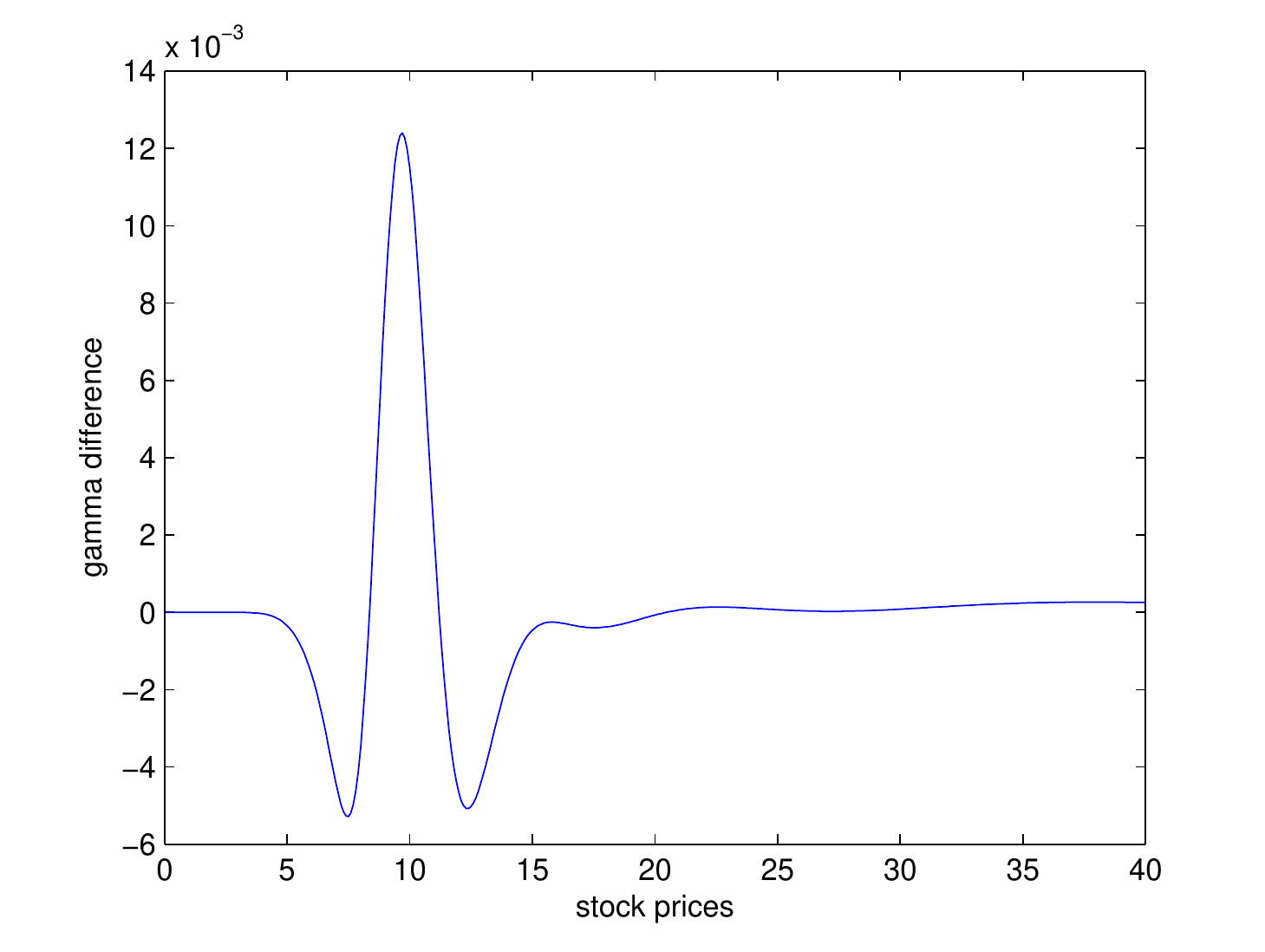}
\end{center}
\caption{Comparison of Gamma of our approach and the true values under
  the Black-Scholes model.  Model parameters:
  $t=0.5,K=20,\sigma=0.5,r=10\%$. Basepoint: $z=x$.  The left graph
  plots the gamma computed by our method and the true gamma. The right
  graph plots their difference. Note that in the second figure the
  scale is $10^{-3}$. The x-axis is scaled by 10, that is, the label
  400 means the stock price is 40.}\label{figure gamma}
\end{figure}

\section{Option pricing with long maturity: the bootstrap
scheme}\label{bootstrap}

The \NAME\ gives an asymptotic expansion of the Green function {\em in
  the limit $t\to 0$}.  Therefore, its accuracy is in priciple limited
to times to maturity $t$ relatively small. For long maturity options,
we expect the error to be possibly large. In this section, we shall
introduce a bootstrap strategy to price options with a long maturity
time. The scheme is based on the properties of the solution
operator. Let us illustrate the bootstrap in the time independent
case. In this case, we recall that the solution operator forms a
semigroup. The semigroup property then gives that
\begin{equation}
  e^{tL}=\left (e^{\frac{t}{n}L}\right )^n, \qquad \forall n \in \NN.
\end{equation}
Then, if $n$ is sufficiently large, \
\ $e^{\frac{t}{n}L}$ will be accurately approximated by our method.

We next describe the bootstrap scheme, which can be rigorously
justified at least for the case of strongly elliptic operators ($a$
bounded away from zero) by the error analysis in \cite{CCMN}.  In the
bootstrap scheme, we use \ $\left (\cG^{[n]}_{t/n}\right )^n$ to
approximate $e^{tL}$, where as before we denote the approximate
solution operator by its kernel \ $\cG^{[n]}_t$.  Suppose now
$\cG^{[n]}_{t/n}$ is the second order approximation, then the error is
in the order \ $O(\left ( \frac{t}{n} \right )^{3/2})$. Because there
are $n$ steps in the bootstrap scheme, the total error is in the order
of
\begin{equation*}
  O\big( ( t/n)^{3/2}\big)\times n\ =\
  O\big(t^{3/2}/\sqrt{n}\big),
\end{equation*}
and consequently, for $t$ fixed, it becomes smaller and smaller as $n$
increases.  A similar analysis shows that the bootstrap strategy with
the first order approximation does not improve accuracy, given that in
this case the error at each step is \ $O(t/n)$, so the total error
after $n$ steps is
\begin{equation*}
  O\big( t/n \big)\times n\ =\
  O(t ),
\end{equation*}
which does not converge to zero as $n\to \infty$.

We numerically tested this scheme for both the option prices and the
Greeks.  In the bootstrap scheme, closed-form approximate solutions
are not available after the first time step, since we integrate the
aproximate Green's function against an expression of the form
\eqref{2nd closed form}, which contains error functions. Therefore, we
must integrate numerically and introduce an additional error due to
the numerical quadrature. This error can be controlled and made
lower-order by choosing the space discretization step small enough. A
further error, which can also be made lower-order, comes from the
truncation of the integration at large $x$.

In the first simulation, we used the Black-Scholes-Merton model. and
set the parameters as time to maturity $t=1$ (one year), strike
$K=20$, risk-free interest rate $r=10\%$, and volatility
$\sigma=0.5$. The left graph of Figure \eqref{figure bootstrap}
displays the error of the first-order-closed form solution, the
second-order closed-form solution, the first-order approximation
with bootstrap, and the second-order approximation with bootstrap.
We truncate the half line $(0,+\infty)$ at 200, and fix the number
of the bootstrap steps as $n=10$, that is the time step is $\Delta
t=0.1$.  We choose the space discretization $\Delta x=0.1$. {From}
the graphs, it is clear that the second-order approximation greatly
improves the accuracy compared with the first-order approximation.
The bootstrap scheme with the second-order approximation reduces the
error even further as expected (See Table \ref{table
bootstrap} for a quantitative error analysis).  As predicted, on
the other hand the bootstrap scheme with the first order
approximation introduces an extra error.

We also notice that around $S=40$ (the label 400 in the graphs) the
error with the second-order bootstrap tends to increase, an effect of
the truncation error. To verify it, we truncate the half line at
400. The right graph of Figure \eqref{bootstrap} shows that the error
does not tend to increase. We also tested the cases when the time to
maturity is two and five years, obtaining similar results.

\begin{figure}
\begin{center}
\includegraphics[width=2.2in]{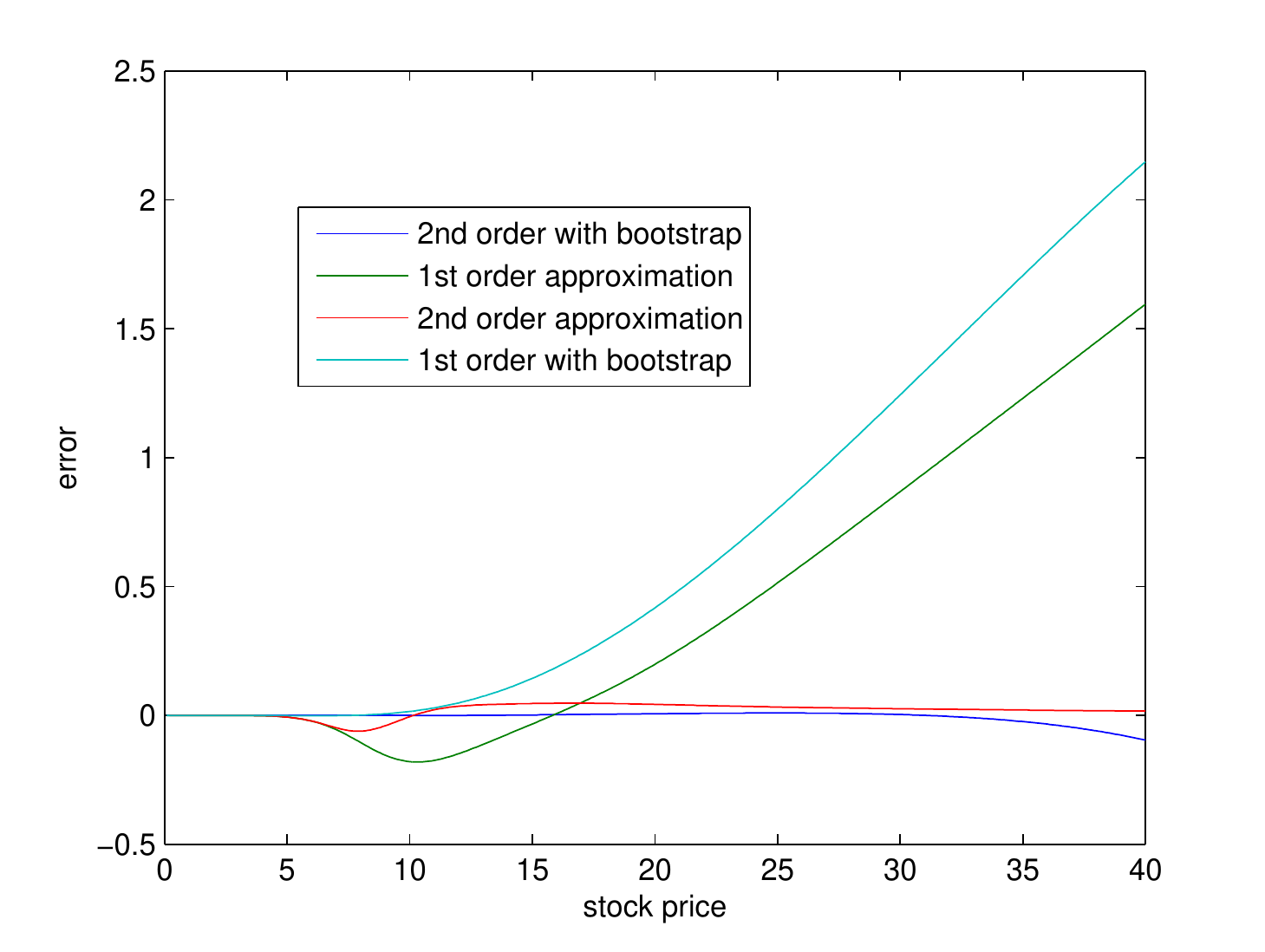}
\includegraphics[width=2.2in]{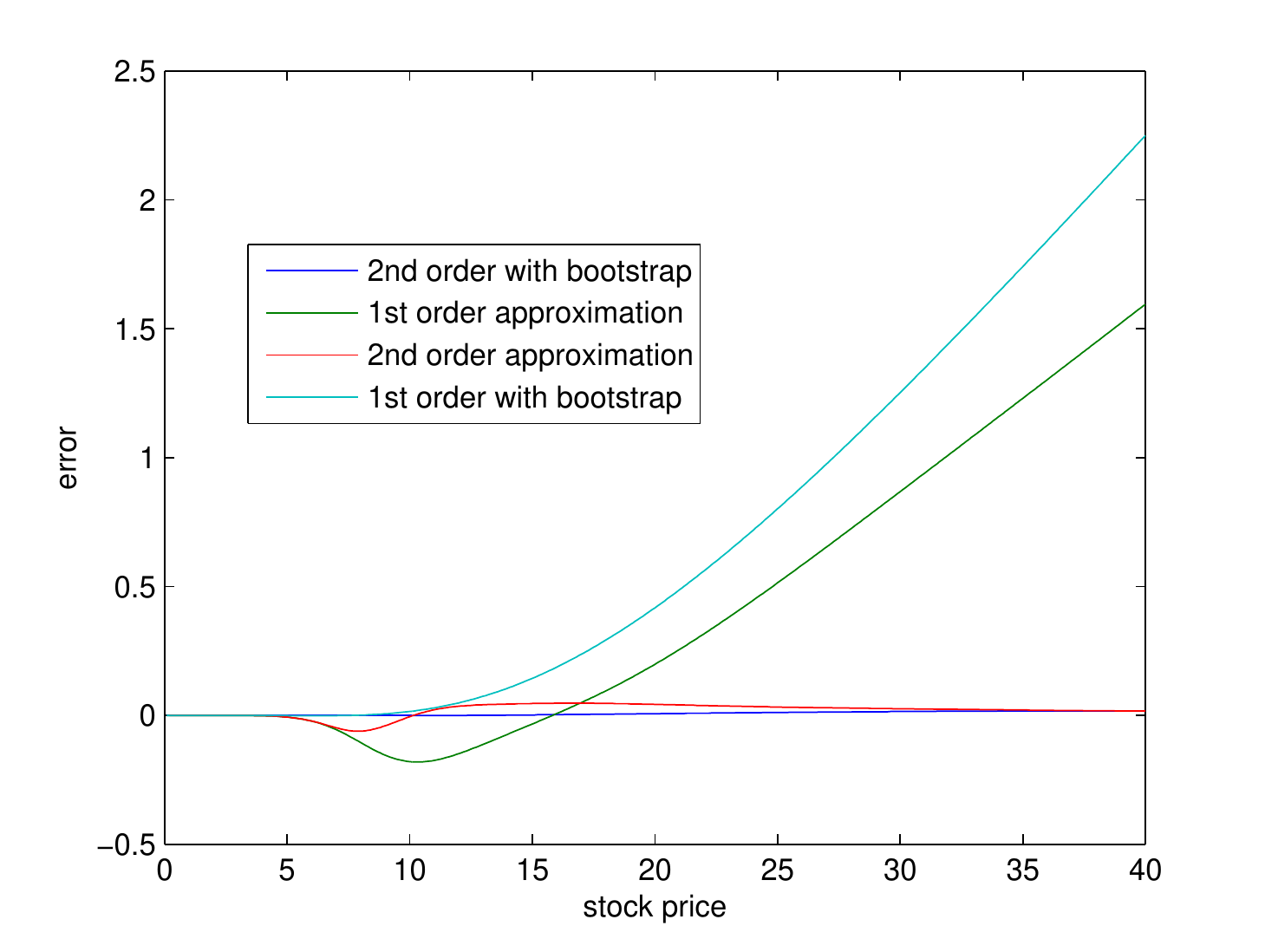}
\end{center}\caption{Comparison of our first and second order approximation with or without bootstrap.
Each curve is the difference with respect to the Black-Scholes
formula. Parameters: $t=1,K=20,r=10\%,\sigma=0.5$. Basepoint: $z=x$.
In our numerical integration, we truncate $(0,\infty)$ at 200 for
the left picture, and at 400 for the right one. Note that the x-axis
scaled by 10, i.e. the label 400 means the stock price is
40.}\label{figure bootstrap}
\end{figure}

To give a sense on how accurate our bootstrap scheme with the second
order approximation is for large time to maturity, we repeat previous
numerical simulation for difference times, and measure the error in
$L^\infty (0,40)$.  We recall that the trike price is at $S=20$, we
are taking a symmetric interval around it, and that the number of
bootstrap steps is fixed at 10. We report the errors in Table
\ref{table bootstrap}.

\begin{table}
\begin{tabular}{|l|l|l|l|l|l|l|}\hline
t& 3&2 & 1 & 0.5 & 0.2 & 0.1 \\\hline
 error &0.0268&0.0379&0.0177&0.0038&4.3682e-004&3.5703e-005\\\hline
\end{tabular}
\medskip
\caption{Errors of the bootstrap scheme for different times under
the Black-Scholes-Merton model. Number of bootstrap steps is fixed
as 10. Parameters: $K=20, \sigma=0.5, r=10\%$. Errors are measured
in $L^\infty (0,40)$, and the benchmark is the Black-Scholes formula
} \label{table bootstrap}
\end{table}

As predicted, we can increase the number of bootstrap steps to obtain
arbitrary accuracy in the aproximation.  Furthermore, for relatively
large $t$ the number of bootstrap steps should be correspondingly
large, so that the compound error from each bootstrap step is under
control at the end. For example, in our tests when $t\geq 4$, it is
not enough to reduce the error by bootstrapping with $10$ steps. Using
only $10$ steps in this case, in fact, introduces additional
errors. For more detail, see \cite{CMN}.

In order to eliminate the effect of the truncation error, we shall
work with a butterfly option in the rest of this section.
Mathematically, a butterfly option corresponds to an intial pay-off
given by a hat function, Figure \eqref{butterfly}. Our method gives
closed-form solution for butterfly options as well, by linearity.
Figure \eqref{bootstrapButterfly} shows the errors of a butterfly
option within the Black-Scholes-Merton model with $K=20$, $K_1=15$,
$K_2=25$ obtained by our first order and second order approximation
with or without bootstrap. The benchmark is the true solution. The
parameters we were using are the same as we mentioned before. Again,
we truncate the half line at $200$. For a butterfly option, the
truncation error is clearly very small, given that the data is
compactly supported (Figure \eqref{bootstrapButterfly}). For the
second order approximation with a bootstrap scheme, the error is
almost zero. It is in the scale of $10^{-3}$, while without the
bootstrap the error is of the scale $10^{-2}$. This coincides with our
theoretical results.

We can also run the simulation as in Table \ref{table bootstrap}, and
we find comparable results.

\begin{figure}
\flushleft
\hspace*{-.3in}
\begin{minipage}[t]{.5\textwidth}
\includegraphics[width=2.6in]{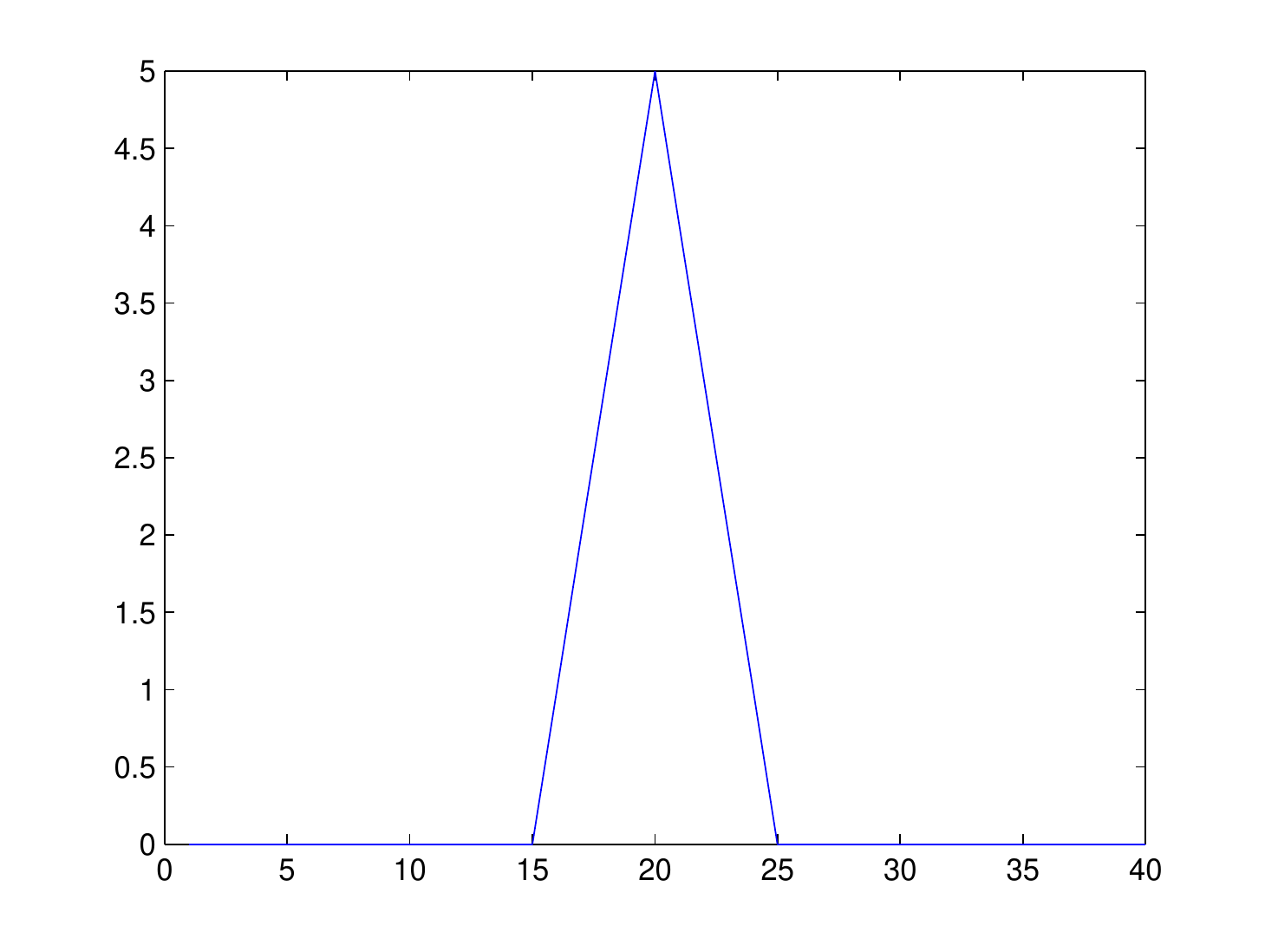}
\hspace*{-.4in}\caption{Butterfly option payoff.}\label{butterfly}
\end{minipage}
\begin{minipage}[t]{.8\textwidth}
\includegraphics[width=2.6in]{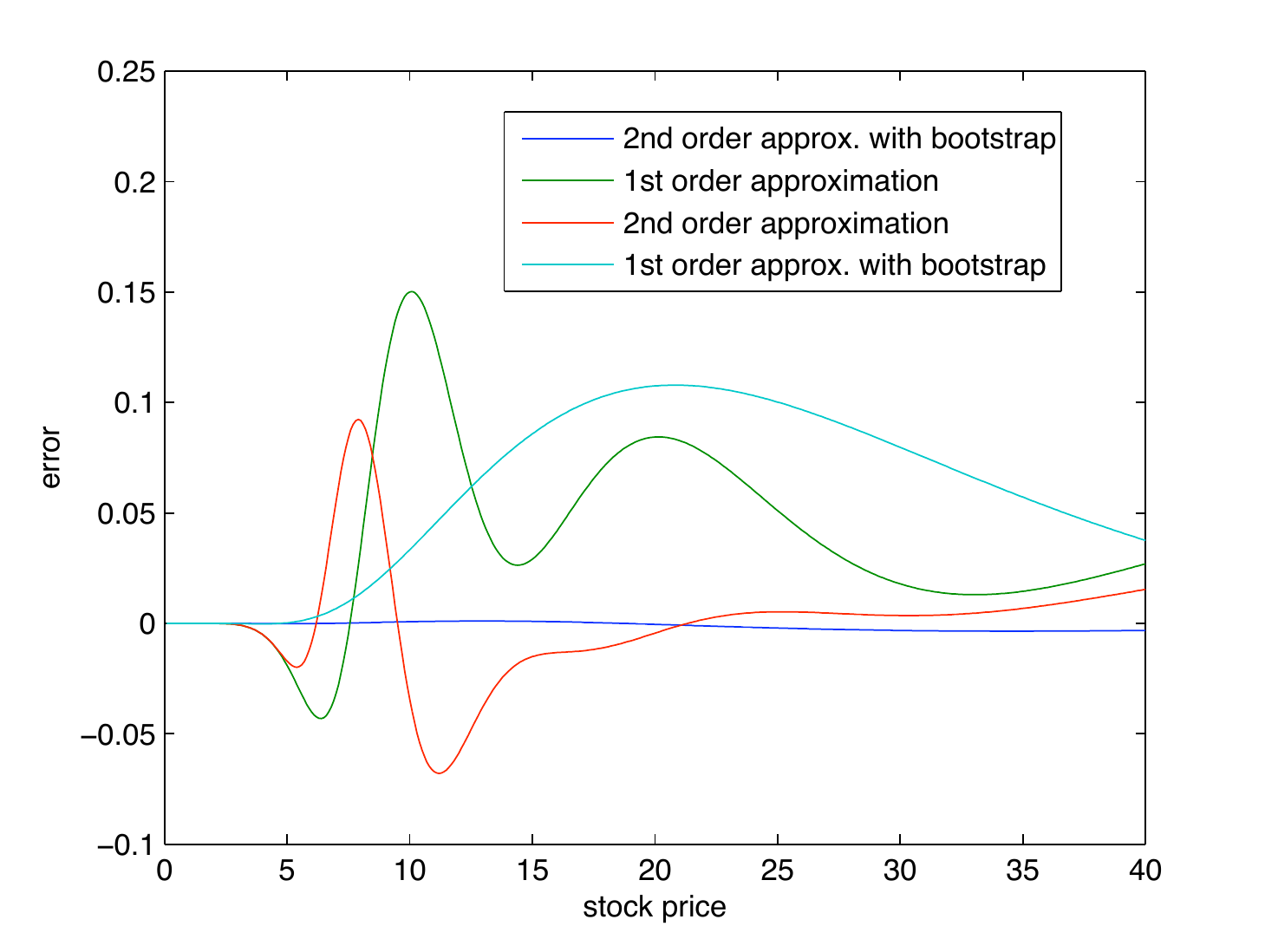}
\hspace{-.4in}\caption{comparison of our first order and second order
approximation with or without bootstrap for a butterfly option when
$t=1,r=10\%,\sigma=0.5$. In our numerical integration we truncate
the half line at 200. Note that the x-axis is scaled by 10, i.e. the
label 400 means the stock price is 40.}\label{bootstrapButterfly}
\end{minipage}\\[20pt]
\end{figure}

We conclude by discussing the bootstrap scheme for the Greeks.
Directly using the closed-form approximation formula to compute the
Greeks for very long maturity time ($t>>1$) is not advisable.  In
fact, our closed-form approximation for call option oscillates near
the strike price, and the oscillation grows with the time to expiry,
as the overall error grows. The appearance of the oscillation is due
to the discontinuity of first derivative of the pay-off function at
the strike price.  This phenomenon is clearly visible for butterfly
options, where the first derivative of the payoff function has three
discontinuities, Figure \eqref{bootstrapButterfly}.  This oscillation
is amplified in the calculation of Greeks. The bootstrap scheme
reduces this oscillation dramatically.

For the numerical simulation, we choose the same parameters as those
in Section \eqref{greeks}.  The small time step ensures very good
error control at each time step. Also, we minimize the truncation
error as before by truncating the integral at $x=400$ and comparing
the approximations only on the interval $[0,40]$ near the strike price
$k=20$. The left graph of Figure \eqref{figure delta bootstrap} plots
the true delta and our approximation in the same picture, and the
right one plots the difference between these two curves, which shows
that the difference between the true value and our approximation is in
the order of $10^{-4}$ with the biggest error around $3.3\times
10^{-4}$. Thus our approximation is quite accurate. For the gamma, we
obtain similar results, see Figure \ref{figure gamma bootstrap}. The
difference is in the order of $10^{-5}$, and the biggest error is
around $5.6\times 10^{-5}$. In both cases, there are no oscillations
on the same scale of the solution.

\begin{figure}
\begin{center}
\includegraphics[width=2.2in]{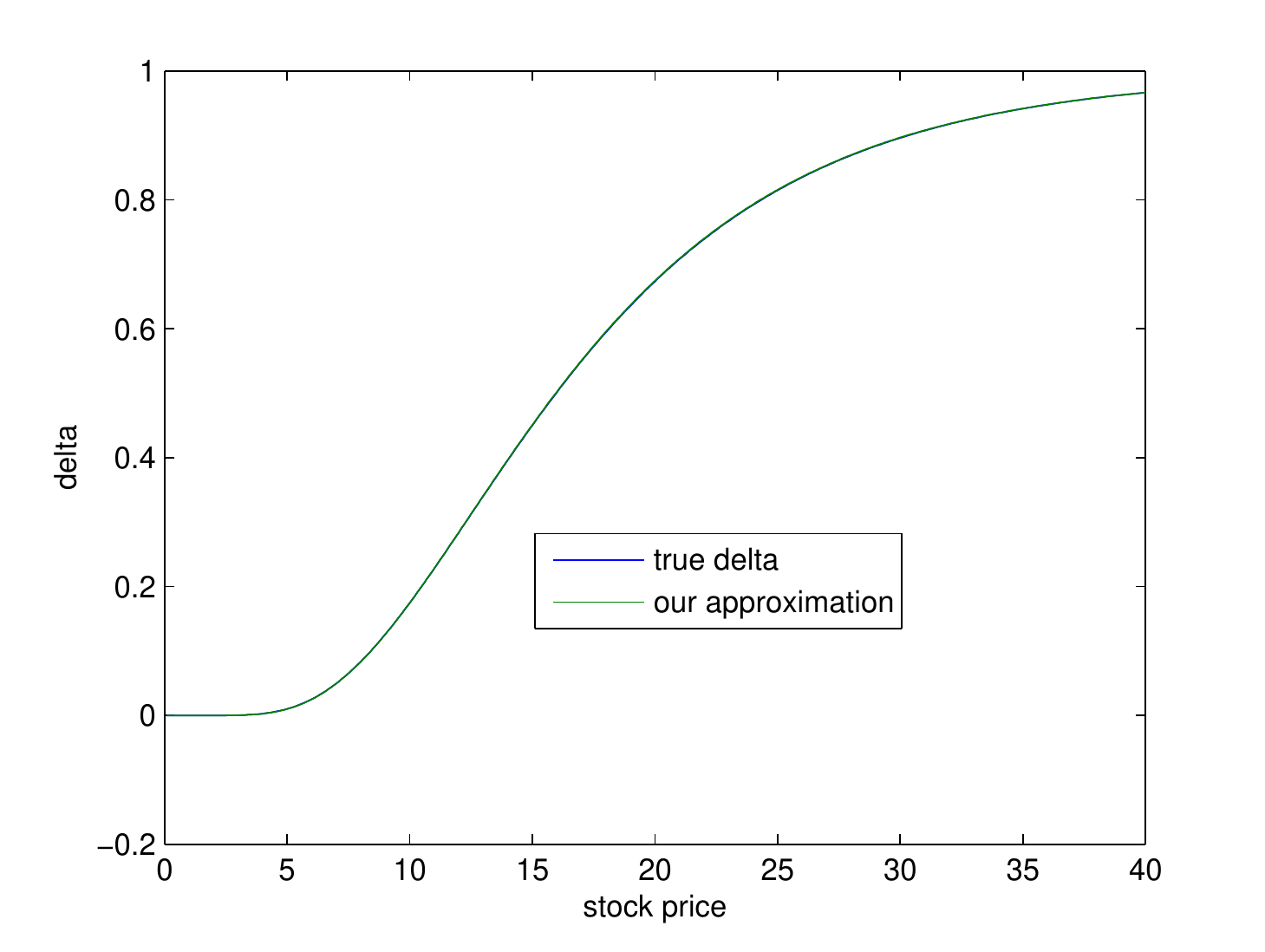}
\includegraphics[width=2.2in]{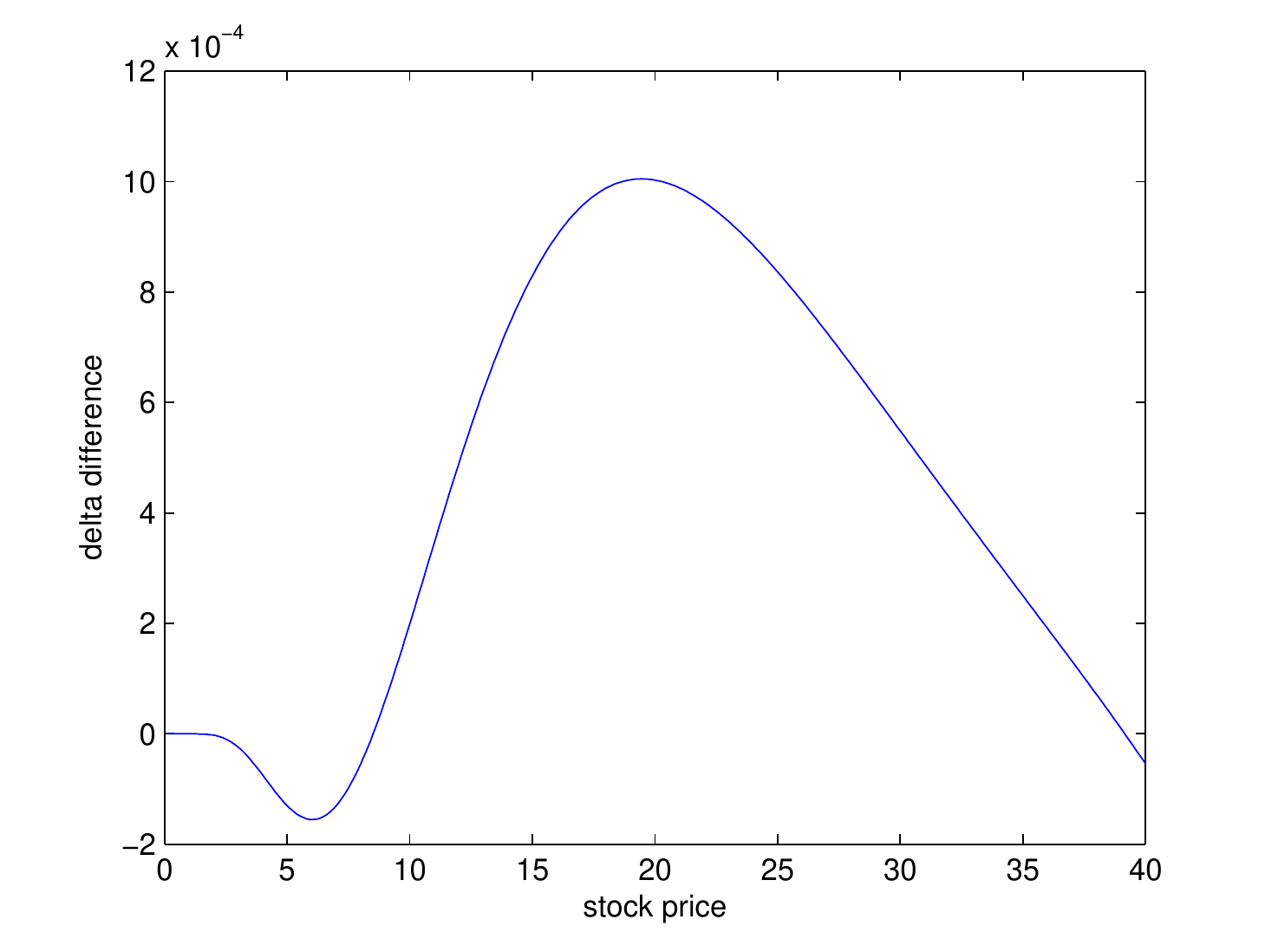}
\end{center} \caption{Comparison of the delta computed by our bootstrap scheme( 2nd order approximation) with
the true values under the Black-Scholes model. Model parameters:
$t=0.5,K=20,\sigma=0.5,r=10\%$. Basepoint: $z=x$. The left graph
plots the delta computed by our bootstrap method and the true delta.
The right graph plots their difference. In our numerical
integration, we truncate the half line at 400. Note that in the
second figure the scale is $10^{-4}$. The x-axis scaled by 10, i.e.
the label 400 means the stock price is 40. }\label{figure delta
bootstrap}
\end{figure}

\begin{figure}
\begin{center}
\includegraphics[width=2.2in]{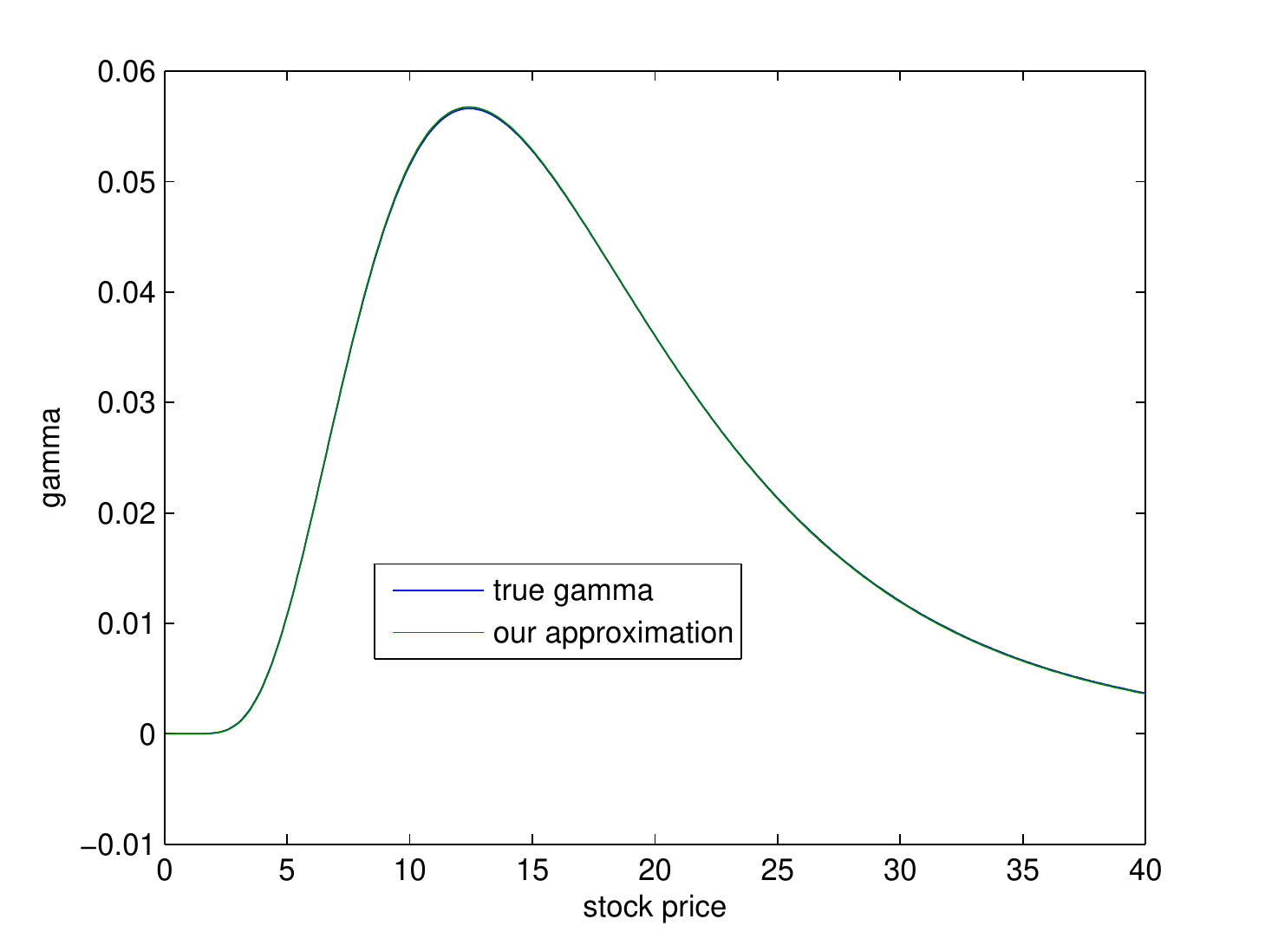}
\includegraphics[width=2.2in]{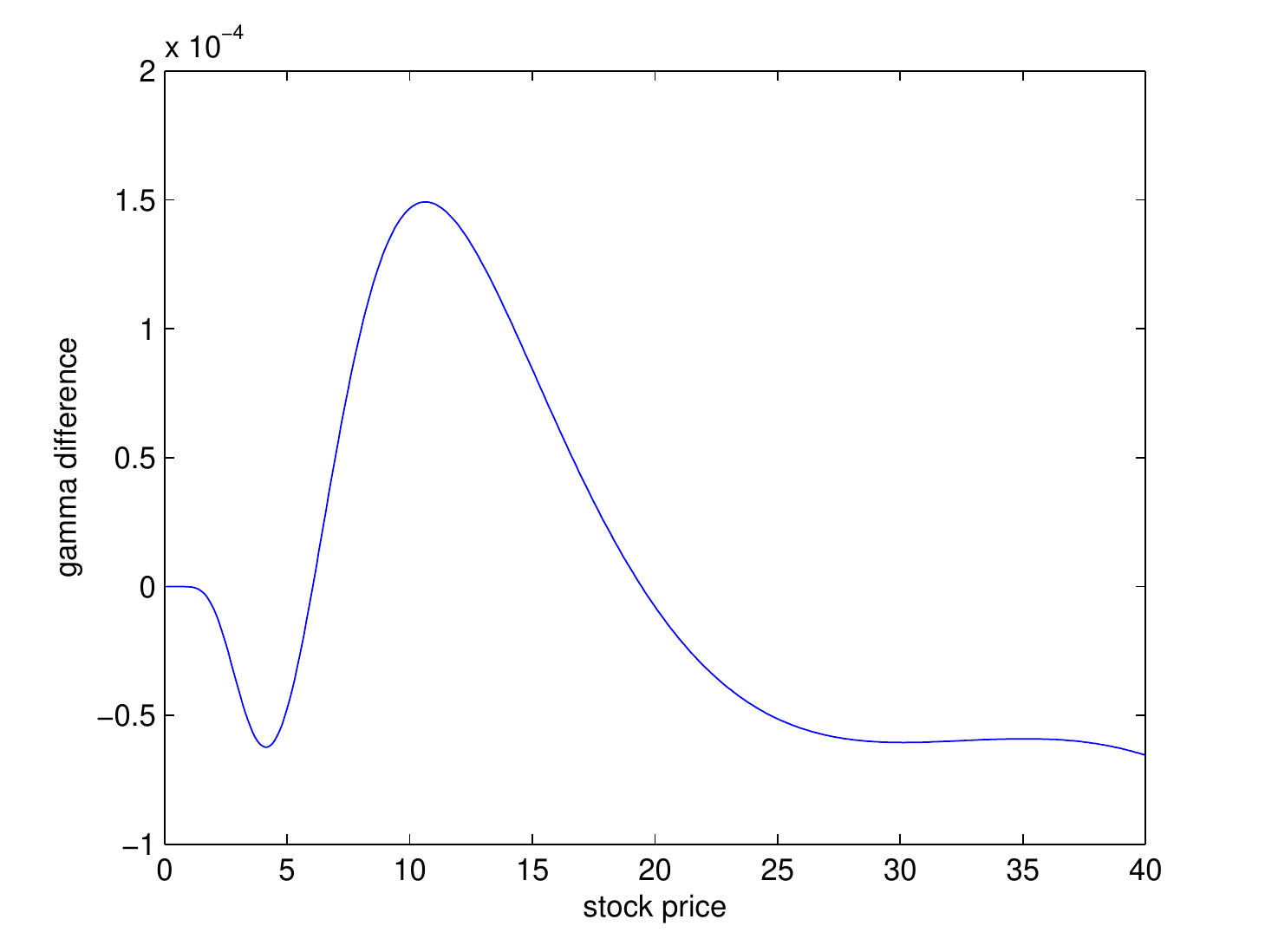}
\end{center}\caption{Comparison of the gamma computed by our bootstrap scheme( 2nd order approximation) with
the true values under the Black-Scholes model. Model parameters:
$t=0.5,K=20,\sigma=0.5,r=10\%$. Basepoint: $z=x$. The left graph
plots the gamma computed by our bootstrap method and the true gamma.
The right graph plots their difference. In our numerical
integration, we truncate the half line at 400. Note that in the
second figure the scale is $10^{-5}$. The x-axis is scaled by 10,
i.e. the label 400 means the stock price is 40.}\label{figure gamma
bootstrap}
\end{figure}




\end{document}